\documentclass[journal,comsoc]{IEEEtran}
\usepackage[T1]{fontenc}
\usepackage{graphicx}
\usepackage{amssymb}
\usepackage{amsmath}
\usepackage{cite}
\usepackage{subfigure}
\usepackage{mathrsfs}
\usepackage[displaymath,mathlines]{lineno}
\usepackage{color}
\usepackage{tabulary}
\usepackage{multirow}
\usepackage{algpseudocode}
\usepackage{algorithm,algpseudocode}
\usepackage{algorithmicx}
\usepackage{pbox}
\usepackage{multicol}
\usepackage{lipsum}
\usepackage{cleveref}
\usepackage{psfrag}
\usepackage{url}
\usepackage{bbm}
\usepackage{epstopdf}
\usepackage{pgf}
\usepackage{amsthm}

\newtheorem{theorem}{Theorem}
\newtheorem{lemma}{Lemma}
\newtheorem{corollary}{Corollary}
\newtheorem{proposition}{Proposition}

\newtheorem{remark}{Remark}

\def\BibTeX{{\rm B\kern-.05em{\sc i\kern-.025em b}\kern-.08em
    T\kern-.1667em\lower.7ex\hbox{E}\kern-.125emX}}

\begin{document}
\vspace{0 cm}
\title{
   Joint Power Allocation and User Association Optimization for Massive MIMO Systems}
\author{
         Trinh Van Chien \textit{Student Member}, \textit{IEEE}, Emil Bj\"{o}rnson, \textit{Member}, \textit{IEEE}, and Erik G. Larsson, \textit{Fellow}, \textit{IEEE}
\thanks{
        The authors are with the Department of Electrical
    Engineering (ISY), Link\"{o}ping University, 581 83 Link\"{o}ping,
    Sweden
        (email: trinh.van.chien@liu.se; emil.bjornson@liu.se; erik.g.larsson@liu.se).
}
\thanks{
       This paper was supported by the European Union's Horizon 2020 research and innovation programme under grant agreement No 641985 (5Gwireless). It was also supported by ELLIIT and CENIIT.
}
\thanks{
       Parts of this paper were presented at IEEE ICC 2016 \cite{Chien2016a}.
}
}

\markboth{}
        {}

\maketitle

\vspace{-1.2cm}
\begin{abstract}
This paper investigates the joint power allocation and user association problem in multi-cell Massive MIMO (multiple-input multiple-output) downlink (DL) systems. The target is to minimize the total transmit power consumption when each user is served by an optimized subset of the base stations (BSs), using non-coherent joint transmission. We first derive a lower bound on the ergodic spectral efficiency (SE), which is applicable for any channel distribution and precoding scheme. Closed-form expressions are obtained for Rayleigh fading channels with either maximum ratio transmission (MRT) or zero forcing (ZF) precoding. From these bounds, we further formulate the DL power minimization problems with fixed SE constraints for the users. These problems are proved to be solvable as linear programs, giving the optimal power allocation and BS-user association with low complexity. Furthermore, we formulate a max-min fairness problem which maximizes the worst SE among the users, and we show that it can be solved as a quasi-linear program. Simulations manifest that the proposed methods provide good SE for the users using less transmit power than in small-scale systems and the optimal user association can effectively balance the load between BSs when needed. Even though our framework allows the joint transmission from multiple BSs, there is an overwhelming probability that only one BS is associated with each user at the optimal solution.
\end{abstract}


\section{Introduction} \label{Sec:Introduction}

The exponential growth in wireless data traffic and number of wireless devices cannot be sustained by the current cellular network technology. The fifth generation (5G) cellular networks are expected to bring thousand-fold system capacity improvements over contemporary networks, while also supporting new applications with massive number of low-power devices, uniform coverage, high reliability, and low latency \cite{Hossain2014, Wang2014}. These are partially conflicting goals that might need a combination of several new radio concepts; for example, Massive MIMO \cite{Marzetta2010a}, millimeter wave communications \cite{Rappaport2013}, and device-to-device communication \cite{Tehrani2014}.

Among them, Massive MIMO, a breakthrough technology proposed in \cite{Marzetta2010a}, has gained lots of attention recently \cite{Ngo2013a, Larsson2014a, Gao2015a, Bjornson2016b}. It is considered as an heir of the MIMO technology since its scalability can provide very large multiplexing gains, while previous single-user and multi-user MIMO solutions have been severely limited by the channel estimation overhead and unfavorable channel properties. In Massive MIMO, each BS is equipped with hundreds of antennas and serves simultaneously tens of users. Since there are many more antennas than users, simple linear processing techniques such as MRT or ZF, are close to optimal. The estimation overhead is made proportional to the number of users by sending pilot signals in the uplink (UL) and utilizing the channel estimates also in the DL by virtue of time-division duplex (TDD).

Because $80 \%$ of the power in current networks is consumed at the BSs \cite{Auer2011a}, the BS technology needs to be redesigned to reduce the power consumption as the wireless traffic grows. Many researchers have investigated how the physical layer transmissions can be optimized to reduce the transmit power, while maintaining the quality-of-service (QoS); see \cite{Auer2011a, Rashid1998b, Stridh2006a, Sun2015, Bjornson2013d, Li2015} and references therein. In particular, the precoding vectors and power allocation were jointly optimized in \cite{Rashid1998b} under perfect channel state information (CSI). The algorithm was extended in \cite{Stridh2006a} to also handle the BS-user association, which is of paramount importance in heterogeneously deployed networks and when the users are heterogeneously distributed. However, \cite{Stridh2006a} did not include any power constraints at the BSs, which could lead to impractical solutions. In contrast, \cite{Sun2015} showed that most joint power allocation and BS-user association problems with power constraints are NP-hard. The recent papers \cite{Bjornson2013d, Li2015} consider a relaxed problem formulation where each user can be associated with multiple BSs and show that these problems can be solved by convex optimization.

The papers \cite{Rashid1998b, Stridh2006a, Sun2015, Bjornson2013d, Li2015} are all optimizing power with respect to the small-scale fading, which is very computationally demanding since the fading coefficients change rapidly (i.e., every few milliseconds). It is also unnecessary to compensate for bad fading realizations by spending a lot of power on having a constant QoS, since it is typically the average QoS that matters to the users. In contrast, the small-scale fading has negligible impact on Massive MIMO systems, thanks to favorable propagation \cite{Ngo2014a}, and closed-form expressions for the ergodic SE are available for linear precoding schemes \cite{Ngo2013a}. The power allocation can be optimized with respect to the slowly varying large-scale fading instead \cite{Larsson2014a}, which makes advanced power control algorithms computationally feasible. A few recent works have considered power allocation for Massive MIMO systems. For example, the authors in \cite{Zhao2013} formulated the DL energy efficiency optimization problem for the single cell Massive MIMO systems that takes both the transmit and circuit powers into account. The paper \cite{Guo2014a} considered optimized user-specific pilot and data powers for given QoS constraints, while \cite{Victor2015b} optimized the max-min SE and sum SE. None of these papers have considered the BS-user association problem. 

Massive MIMO has demonstrated high energy efficiency in homogeneously loaded scenarios \cite{Ngo2013a}, where an equal number of users are preassigned to each BS.  At any given time, the user load is typically heterogeneously distributed, such that some BSs have many more users in their vicinity than others. Large SE gains are often possible by balancing the load over the network \cite{Ye2013a, Andrews2014a}, using some other user association rule than the simple maximum signal-to-noise ratio (max-SNR) association. Instead of associating a user with only one BS, coordinated multipoint (CoMP) methods can be used to let multiple BSs jointly serve a user \cite{Boldi2011a}. This can either be implemented by sending the same signal from the BSs in a coherent way, or by sending different simultaneous signals in a non-coherent way. However, finding the optimal association is a combinatorial problem with a complexity that scales exponentially with the network size \cite{Andrews2014a}.  Such association rules are referred to as a part of CoMP joint transmission and have attracted significant interest because of their potential to increase the achievable rate \cite{ Boldi2011a, Ye2013a}. While load balancing is a well-studied problem for heterogeneous multi-tier networks, the recent works \cite{Bethanabhotla2016a, Liu2015a, Bjornson2013e} have shown that large gains are possible also in Massive MIMO systems. From the game theory point of view, the author in \cite{Bethanabhotla2016a} proposed a user association approach to maximize the SE utility while taking pilot contamination into account. Apart from this, \cite{Liu2015a} considered the sum SE maximization of a network where one user is associated with one BS. We note that \cite{Bethanabhotla2016a, Liu2015a} only investigated user association problems for a given transmit power at the BSs. Different from \cite{Bethanabhotla2016a, Liu2015a}, the total power consumption minimization problems with optimal and sub-optimal precoding schemes were investigated in \cite{Bjornson2013e}.

In this paper we jointly optimize the power allocation and BS-user association for multi-cell Massive MIMO DL systems. Specifically, our main contributions are as follows:
\begin{itemize}
\item We derive a new ergodic SE expression for the scenario when the users can be served by multiple BSs, using non-coherent joint transmission and decoding the received signals in a successive manner. Closed-form expressions are derived for MRT and ZF precoding. 
\item We formulate a transmit power minimization problem under ergodic SE requirements at the users and limited power budget at the BSs. This problem is shown to be a linear program when the new ergodic SE expression for MRT or ZF is used, so the optimal solution is found in polynomial time.
\item The optimal BS-user association rule is obtained from the transmit power minimization problem. This rule reveals how the optimal association depends on the large-scale fading, estimation quality, signal-to-interference-and-noise ratio (SINR), and pilot contamination. Interestingly, only a subset of BSs serves each user at the optimal solution.
\item We consider the alternative option of optimizing the SE targets utilizing max-min SE formulation with user-specific weights. This problem is shown to be quasi-linear and can be solved by an algorithm that combines the transmit power minimization with the bisection method.
\item The effectiveness of our novel algorithms and analytical results are demonstrated by extensive simulations. These show that the power allocation, array gain, and BS-user association are all effective means to decrease the power consumption in the cellular networks. Moreover, we show that the max-min algorithm can provide uniformly great SE for all users, irrespective of user locations, and provide a map that shows how the probability of being served by a certain BS depends on the user location.
\end{itemize}

This paper is organized as follows: Section \ref{Section:System-Model and Achievable Performance} presents the multi-cell Massive MIMO system model and derives lower bounds on the ergodic SE. In Section \ref{Section:Power-Optimization} the transmit power minimization problem is formulated. The optimal solution is obtained in Section \ref{Section:Optimal-Solution} where also the optimal BS-user association rule is obtained, while an algorithm for max-min SE optimization is derived in Section \ref{Section:Max-Min-QoS}. Finally, Section \ref{Section:Numerical-Results} gives numerical results and Section \ref{Section:Conclusion} summarizes the main conclusions.

\textit{Notations:}  We use upper-case bold face letters for matrices and lower-case bold face ones for vectors. $\mathbf{I}_M$ and $\mathbf{I}_K$ are the identity matrices of size $M \times M$ and $K \times K$, respectively. The operator $\mathbb{E} \{ \cdot \}$ is the expectation of a random variable. The notation $ \| \cdot \| $ stands for the Euclidean norm and $\mathrm{tr} ( \cdot)$ is the trace of a matrix. The regular and Hermitian transposes are denoted by $(\cdot)^T$ and $(\cdot)^H$, respectively. Finally, $\mathcal{CN}(.,.)$ is the circularly symmetric complex Gaussian distribution.

\section{System Model and Achievable Performance} \label{Section:System-Model and Achievable Performance}
\begin{figure}[t]
    \centering
      \includegraphics[width=0.35\textwidth]{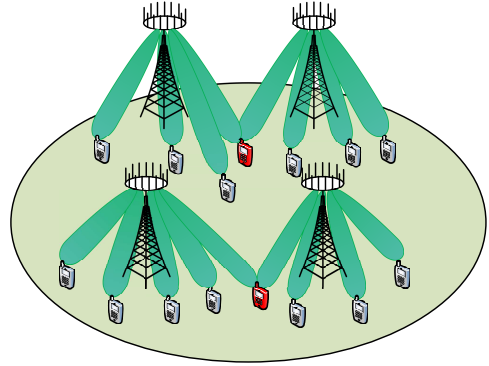}
    \caption{A multiple-cell Massive MIMO DL system where users can be associated with more than one BS (e.g., red users). The optimized BS subset for each user is obtained from the proposed optimization problem.}
    \label{fig:MassiveMIMOSystems}
\end{figure}

A schematic diagram of our system model is shown in Fig.~\ref{fig:MassiveMIMOSystems}. We consider a Massive MIMO system with $L$ cells. Each cell comprises a BS with $M$ antennas. The system serves $K$ single antenna users in the same time-frequency resource. Note that each user is conventionally associated and served by only one of the BSs. However, in this paper, we optimize the BS-user association and investigate when it is preferable to associate a user with multiple BSs. Therefore, the users are numbered from $1$ to $K$ without having predefined cell indices.

We assume that the channels are constant and frequency-flat in a coherence interval of length $\tau_c$ symbols and the system operates in TDD mode. In detail, $\tau_p$ symbols are used for channel estimation, so the remaining portion of a coherence block including $ \tau_c - \tau_p$ symbols are dedicated for the data transmission. In the UL, the received baseband signal $\mathbf{y}_l \in \mathbb{C}^{M}$ at BS $l$, for $l=1, \ldots, L,$ is modeled as
\begin{equation} \label{eq: UL-Recieved-Signal-Vector}
\mathbf{y}_{l} = \sum_{t=1}^{K} \mathbf{h}_{l,t} \sqrt{p_t} x_t + \mathbf{n}_l,
\end{equation}
where $p_t$ is the transmit power of user $t$ assigned to the normalized transmit symbol  $x_t$ with $\mathbb{E} \{| x_t |^2\} =1$. At each BS,  the receiver hardware is contaminated by additive noise $\mathbf{n}_l \sim \mathcal{CN}( \mathbf{0}, \sigma_{ \mathrm{UL} }^2 \mathbf{I}_M )$. The vector $\mathbf{h}_{l, t} $ denotes the channel between user $t$ and BS $l$. In this paper, we consider uncorrelated Rayleigh fading channels, meaning that the channel realizations are independent between users, BS antennas and between coherence intervals. Mathematically, each channel vector $\mathbf{h}_{l, t} $, for $t= 1, \ldots, K$, is a realization of the circularly symmetric complex Gaussian distribution
\begin{equation} \label{eq: Rayleigh-Fading}
\mathbf{h}_{l,t} \sim \mathcal{CN} ( \mathbf{0}, \beta_{l,t} \mathbf{I}_M ).
\end{equation}
The variance $\beta_{l,t}$ describes the large-scale fading which, for example, symbolizes the attenuation of signals due to diffraction around large objects such as high buildings and due to propagation over a long distance between the BS and user. Let us define the channel matrix $\mathbf{H}_l =[ \mathbf{h}_{l,1}, \ldots ,\mathbf{h}_{l,K}]  \in \mathbb{C}^{M \times K}$,  the diagonal power matrix $\mathbf{P} = \mathrm{diag} (p_1,\ldots,p_K) \in \mathbb{C}^{ K \times K}$, and the useful signal vector $\mathbf{x}_{l} = [ x_{l,1},\ldots, x_{l,M}]^T  \in \mathbb{C}^{M}$. Thus, the UL received signal at BS $l$ in \eqref{eq: UL-Recieved-Signal-Vector} can be written as
\begin{equation} \label{eq: UL-Recieved-Signal-Matrix}
\mathbf{y}_l = \mathbf{H}_l \mathbf{P}^{1/2} \mathbf{x}_l + \mathbf{n}_l.
 \end{equation}
 
Each BS in a Massive MIMO system needs CSI in order to make efficient use of its antennas; for example, to coherently combine desired signals and reject interfering ones. BSs do not have CSI a priori, which calls for CSI estimation from UL pilot signals in every coherence interval.

\subsection{Uplink Channel Estimation}
The pilot signals are a part of the UL transmission. We assume that user $k$ transmit the pilot sequence $\pmb{\phi}_k$ of length $\tau_p$ symbols described by the UL model in \eqref{eq: UL-Recieved-Signal-Matrix}. 
We let $\mathcal{P}_k  \subset \{ 1, \ldots, K \}$ denote the set of user indices, including user $k$, that use the same pilot sequence as user $k$. Thus, the pilot sequences are assumed to be mutually orthogonal such that
\begin{equation}
\pmb{\phi}_t^H \pmb{\phi}_k = 
\begin{cases}
0, \; t \notin \mathcal{P}_k,\\
\tau_p, \; t \in \mathcal{P}_k .
\end{cases}
\end{equation}
The received pilot signal $\mathbf{Y}_l \in \mathbb{C}^{M \times \tau_p}$ at BS $l$ can be expressed as
\begin{equation} \label{eq: UL-Recieved-Pilot-Matrix}
\mathbf{Y}_l = \mathbf{H}_l \mathbf{P}^{1/2} \pmb{ \Phi }^H + \mathbf{N}_l,
\end{equation}
where the $\tau_p \times K$ pilot matrix $\pmb{\Phi} =  [\pmb{\phi}_1, \ldots , \pmb{\phi}_K]$ and $\mathbf{N}_l \in \mathbb{C}^{M \times \tau_p}$ is Gaussian noise with independent entries having the distribution $\mathcal{CN} (0, \sigma_{\mathrm{UL}}^2 )$. Based on the received pilot signal \eqref{eq: UL-Recieved-Pilot-Matrix} and assuming that the BS knows the channel statistics, it can apply minimum mean square error (MMSE) estimation \cite{Kay1993a} to obtain a channel estimate of $\mathbf{h}_{l,k}$ as shown in the following lemma.
 \begin{lemma} \label{Lemma:ChannelEstimate}
BS $l$  can estimate the channel to user $k$ using MMSE estimation from the following equation,
\begin{equation}
\begin{split}
\mathbf{Y}_l \pmb{\phi}_k  = \mathbf{H}_l \mathbf{P}^{1/2} \pmb{ \Phi }^H \pmb{\phi}_k  + \mathbf{N}_l \pmb{\phi}_k =  \tau_p \sum_{t' \in \mathcal{P}_k} \sqrt{p_{t'}}  \mathbf{h}_{l,t'} + \tilde{ \mathbf{n} }_{l,k},
\end{split}
\end{equation}
where $\tilde{ \mathbf{n} }_{l,k} = \mathbf{N}_l \pmb{\phi}_k \sim \mathcal{CN}( \mathbf{0}, \tau_p \sigma_{\mathrm{UL}}^2 \mathbf{I}_M )$. The MMSE estimate $\hat{ \mathbf{h} }_{l,k}$  of the channel $\mathbf{h}_{l,k}$ between BS $l$ and user $k$ is
\begin{equation}
\hat{\mathbf{h}}_{l,k} = \frac{ \sqrt{p_{k}} \beta_{l,k} }{  \tau_p \sum_{t' \in \mathcal{P}_k } p_{t'} \beta_{l,t'} + \sigma_{ \mathrm{UL} }^2   } \mathbf{Y}_l \pmb{\phi}_k
\end{equation}
and the estimation error is defined as
\begin{equation} \label{eq: Error-Vector}
\mathbf{e}_{l,k} = \hat{\mathbf{h}}_{l,k} - \mathbf{h}_{l,k}.
\end{equation}
Consequently, the channel estimate and the estimation error are independent and distributed as
\begin{align}
\label{eq: Estimate-Channel-Distribution}
\hat{ \mathbf{h} }_{l,k} & \sim \mathcal{CN} \left(  \mathbf{0},  \theta_{l,k} \mathbf{I}_M \right), \\
\label{eq: Estimation-Error-Distribution}
 \mathbf{e}_{l,k} & \sim \mathcal{CN} \left(  \mathbf{0}, \left( \beta_{l,k} - \theta_{l,k} \right)  \mathbf{I}_M  \right),
\end{align}
where
\begin{equation}
\theta_{l,k} = \frac{ p_k  \tau_p \beta_{l,k}^2 }{ \tau_p \sum_{t' \in \mathcal{P}_k } p_{t'} \beta_{l,t'} + \sigma_{ \mathrm{UL} }^2 }.
\end{equation}
\end{lemma}
\begin{proof}  
The proof follows from the standard MMSE estimation of Gaussian random variables \cite{Kay1993a}. 
\end{proof}

In a compact form, each BS $l$ produces a channel estimate matrix $\widehat{ \mathbf{H}}_l = [ \hat{ \mathbf{h} }_{l,1}, \ldots, \hat{ \mathbf{h} }_{l,k}] \in \mathbb{C}^{ M \times K} $  and the mismatch with the true channel matrix $\mathbf{H}_l$ is expressed by the uncorrelated error matrix $ \mathbf{E}_l = [\mathbf{e}_{l,1}, \ldots, \mathbf{e}_{l,K}] \in \mathbb{C}^{ M \times K}$. Lemma \ref{Lemma:ChannelEstimate} provides the statistical properties of the channel estimates that are needed to analyze utility functions like the DL ergodic SE in multi-cell Massive MIMO systems. At this point, we note that the channel estimates of two users $t$ and $k$ in the set $\mathcal{P}_k$ are correlated since they use the same pilot. Mathematically, they are only different from each other by a scaling factor
\begin{equation} \label{eq:Channel_Relationship}
\hat{\mathbf{h}}_{l,k} = \frac{ \sqrt{p_k} \beta_{l,k}}{\sqrt{p_t} \beta_{l,t} } \hat{\mathbf{h}}_{l,t}.
\end{equation}
From the distributions of channel estimates and estimation errors, we further formulate the joint user association and QoS optimization problems, which are the main goals of this paper. One can also analyze the UL performance, but we leave this for future work due to space limitations. 
\vspace*{-0.2cm}
\subsection{Downlink Data Transmission Model} \label{Downlink Data Transmission Model}

Let us denote $\gamma^{\mathrm{DL}}$ as the fraction of the $\tau_c - \tau_p$ data symbols per coherence interval that are used for DL payload transmission, hence $0 < \gamma^{\mathrm{DL}} \leq 1$ and the number of DL symbols is $\gamma^{\mathrm{DL}}(\tau_c - \tau_p)$. We assume that each BS is allowed to transmit to each user but sends a different data symbol than the other BSs. This is referred to as non-coherent joint transmission \cite{Tanbourgi2014, Apelfrojd2014,Bjornson2010c} and it is less complicated to implement than coherent joint transmission which requires phase-synchronization between the BSs.\footnote{This paper investigates whether or not joint transmission can bring substantial performance improvements to Massive MIMO under ideal backhaul conditions. Note that non-coherent joint transmission requires no extensive backhaul signaling, since the BSs send separate data streams and do not require any instantaneous channel knowledge from other cells.} At BS $l$, the transmitted signal $\mathbf{x}_l$ is selected as 
\begin{equation}
\mathbf{x}_l = \sum_{ t =1 }^{K} \sqrt{ \rho_{l,t} } \mathbf{w}_{l,t} s_{l,t}.
\end{equation}
Here the scalar data symbol $s_{l,t}$, which BS $l$ intends to transmit to user $t$, has unit power $\mathbb{E} \{ | s_{l,t }|^2 \} = 1$ and $\rho_{l,t}$ stands for the transmit power allocated to this particular user. In addition, the corresponding linear precoding vector $ \mathbf{w}_{l,t} \in \mathbb{C}^{M}$ determines the spatial directivity of the signal sent to this user.  We notice that user $t$ is associated with BS $l$ if and only if $\rho_{l,t} \neq 0$, and each user can be associated with multiple BSs. We will later optimize the user association and prove that it is optimal to only let a small subset of BSs serve each user. The received signal at an arbitrary user $k$ is modeled as
\begin{equation} \label{eq: Downlink-Signal}
\begin{split}
& y_k 
 = \sum_{i = 1}^{L} \sqrt{ \rho_{i,k} } \mathbf{h}_{i,k}^H \mathbf{w}_{i,k} s_{i,k} + \sum\limits_{i =1 }^{L} \sum_{ \substack{t =1 \\ t \neq k} }^{K} \sqrt{ \rho_{i,t} } \mathbf{h}_{i,k}^H \mathbf{w}_{i,t} s_{i,t} + n_k. 
\end{split}
\end{equation}
The first part in \eqref{eq: Downlink-Signal} is the superposition of desired signals that user $k$ would like to detect. The second part is multi-user interference that degrades the quality of the detected signals. The third part is the additive white noise $n_k \sim \mathcal{CN} (0, \sigma_{ \mathrm{DL} }^2 )$. 

To avoid spending precious DL resources on pilot signaling, we suppose that user $k$ does not have any information about the current channel realizations but only knows the channel statistics. This works well in Massive MIMO systems due to the channel hardening \cite{Bjornson2016b}. User $k$ would like to detect all the desired signals coming from the BSs. To achieve low computational complexity, we assume that each user detects its different data signals sequentially and applies successive interference cancellation \cite{Tse2005a,Li2015}. Although this heuristic decoding method is suboptimal since we make practical assumptions that the BSs have to do channel estimation and have limited power budget, it is amenable to implement and is known to be optimal for example under perfect channel state information. Suppose that user $k$ is currently detecting the signal sent by an arbitrary BS $l$, say $s_{l,k}$, and possesses the detected signals of the $l-1$ previous BSs  but not their instantaneous channel realizations. From these assumptions, a lower bound on the ergodic capacity between BS $l$ and user $k$ is given in Proposition \ref{Proposition: Rate}. 
\begin{proposition} \label{Proposition: Rate}
If user $k$ knows the signals sent to it by the first $l-1$ BSs in the network, then a lower bound on the DL ergodic capacity between BS $l$ and user $k$ is
\begin{equation}
R_{l, k} = \gamma^{\mathrm{DL}} \left( 1 - \frac{\tau_p}{\tau_c} \right) \log_2 \left( 1 + \mathrm{SINR}_{l,k} \right) \quad \textrm{[bit/symbol]},
\end{equation}
where the SINR, $\mathrm{SINR}_{l,k}$, is given as
\begin{equation} \label{eq: SINR_lk}
\frac{ \rho_{l,k} | \mathbb{E} \{ \mathbf{h}_{l,k}^H \mathbf{w}_{l,k} \} |^2 }{  \sum\limits_{i = 1}^{L} \sum\limits_{t =1}^{K} \rho_{i,t} \mathbb{E} \{ | \mathbf{h}_{i,k}^H \mathbf{w}_{i,t} |^2 \} - \sum\limits_{ i=1}^{l}  \rho_{i,k} | \mathbb{E} \{ \mathbf{h}_{i,k}^H \mathbf{w}_{i,k} \} |^2 + \sigma_{ \mathrm{DL} }^2 }.
\end{equation}
\end{proposition}

\begin{proof}
The proof is given in Appendix \ref{Appendix Proposition: Rate}.
\end{proof}

Each user would like to detect all desired signals coming from the $L$ BSs, or at least the ones that transmit with non-zero powers. Proposition \ref{Proposition: Rate} gives hints to formulate a lower bound on the DL ergodic sum capacity of user $k$. We compute this bound by applying the successive decoding technique described in \cite{Tse2005a,Li2015}.  In detail, the user first detects the signal from BS $1$, while the remaining desired signals are treated as interference. From the $2$nd BS onwards, say BS $l$, user $k$ ``knows" the transmit signals of the $l-1$ previous BSs and can partially subtract them from the received signal (using its statistical channel knowledge). It then focuses on detecting the signal $s_{l,k}$ and considers the desired signals from BS $l+1$ to BS $L$ as interference. By utilizing this successive interference cancellation technique, a lower bound on the DL sum SE at user $k$ is provided in Theorem $\ref{Theorem-Lower-Bound-Rate}$.
\begin{theorem} \label{Theorem-Lower-Bound-Rate}
A lower bound on the DL ergodic sum capacity of an arbitrary user $k$ is 
\begin{equation} \label{eq: Sum-Rate-k}
R_k = \gamma^{\mathrm{DL}} \left( 1 - \frac{\tau_p}{\tau_c} \right) \log_2 (1 + \mathrm{SINR}_k ) \quad \textrm{[bit/symbol]},
\end{equation}
where the value of the effective SINR, $\mathrm{SINR}_k$, is given in \eqref{eq: SINR_k}. 
\begin{figure*} [t]
\begin{equation} \label{eq: SINR_k}
\mathrm{SINR}_k = \frac{ \sum\limits_{ i=1}^{L} \rho_{i,k} | \mathbb{E} \{ \mathbf{h}_{i,k}^H \mathbf{w}_{i,k} \} |^2 }{ \sum\limits_{ i =1}^{L}  \rho_{i,k} (\mathbb{E} \{ | \mathbf{h}_{i,k}^H \mathbf{w}_{i,k} |^2 \} - | \mathbb{E} \{ \mathbf{h}_{i,k}^H \mathbf{w}_{i,k} \} |^2 ) + \sum\limits_{ i=1 }^{L} \sum\limits_{ \substack{t =1 \\ t \neq k} }^{K} \rho_{i,t} \mathbb{E} \{ | \mathbf{h}_{i,k}^H \mathbf{w}_{i,t} |^2 \} + \sigma_{ \mathrm{DL} }^2 }.
\end{equation}
\hrulefill
\end{figure*}
\end{theorem}
\begin{proof}
The proof is also given in Appendix \ref{Appendix Proposition: Rate}.
\end{proof}

The sum SE expression provided by Theorem \ref{Theorem-Lower-Bound-Rate} has an intuitive structure. The numerator in \eqref{eq: SINR_k} is a summation of the desired signal power sent to user $k$ over the average precoded channels from each BS. It confirms that all signal powers are useful to the users and that BS cooperation in the form of non-coherent joint transmission has the potential to increase the sum SE at the users. The first term in the denominator represents beamforming gain uncertainty, caused by the lack of CSI at the terminal, while the second term is multi-user interference and the third term represents the additive noise. Even though we assume user $k$  starts to decode the transmitted signal from the BS $1$, the BS numbering has no impact on $\textrm{SINR}_{k}$ in \eqref{eq: SINR_k}. As a result, the SE is not affected by the decoding orders. Besides, both the lower bounds in Proposition \ref{Proposition: Rate} and Theorem \ref{Theorem-Lower-Bound-Rate} are derived independently of channel distribution and precoding schemes. Thus, our proposed method for non-coherent joint transmission in Massive MIMO systems is applicable for general scenarios with any channel distribution, any selection of precoding schemes, and any pilot allocation. Next, we show that the expressions can be computed in closed form under Rayleigh fading channels, if the BSs utilize MRT or ZF precoding techniques.

\subsection{Achievable Spectral Efficiency under Rayleigh Fading} \label{Achievable-Spectral-Efficiency}
We now assume that the BSs use either MRT or ZF to precode payload data before transmission. Similar to \cite{Bjornson2016a}, the precoding vectors are  described as
\begin{equation}  \label{eq: Linear-Precoding-Vector}
\mathbf{w}_{l,k}  = \begin{cases}  \frac{ \hat{\mathbf{h}}_{l,k} }{ \sqrt{ \mathbb{E} \{  \| \hat{ \mathbf{h}}_{l,k}  \|^{2} \}}}, & \mbox{for MRT,} \\
\frac{ \hat{\mathbf{H}}_{l} \mathbf{r}_{l,k} }{\sqrt{ \mathbb{E} \{ \|\hat{\mathbf{H}}_{l} \mathbf{r}_{l,k}\|^{2} \} }}, & \mbox{for ZF,} \end{cases}
\end{equation}
where $\mathbf{r}_{l,k}$ is the $k$th column of matrix $ ( \widehat{\mathbf{H}}_{l}^{H} \widehat{\mathbf{H}}_{l} )^{-1}$. From the above definition, with the condition $M > K$,  ZF precoding could cancel out interference towards users that BS $l$ is not associated with; this precoding was called full-pilot ZF in \cite{Bjornson2016a}. \footnote{The ZF precoding which we are using here is different from the classical one \cite{Ngo2013a}. More precisely, the classical ZF precoding dedicated to BS $l$ can only cancel out interference towards to the users that are associated with this BS.}. Mathematically, ZF precoding yields the following property
\begin{equation} \label{eq: ZF-Property}
\hat{\mathbf{h}}_{l,t}^H \hat{\mathbf{w}}_{l,k} = \begin{cases}  0, &  t \notin \mathcal{P}_k, \\
\frac{ \sqrt{p_t} \beta_{l,t} }{ \beta_{l,k}\sqrt{ p_k \mathbb{E} \{ \|\widehat{\mathbf{H}}_{l} \mathbf{r}_{l,k}\|^{2} \} }}, & t \in \mathcal{P}_k. \end{cases}
\end{equation}
The lower bound on the ergodic SE in Theorem \ref{Theorem-Lower-Bound-Rate} is obtained in closed forms for MRT and ZF precoding as shown in Corollaries \ref{Corollary-MRT-Rate} and \ref{Corollary-ZF-Rate}.
\begin{corollary} \label{Corollary-MRT-Rate}
For Rayleigh fading channels, if the BSs utilize MRT precoding, then the lower bound on the DL ergodic sum rate in Theorem \ref{Theorem-Lower-Bound-Rate} is simplified to
\begin{equation}
R_k^{ \mathrm{MRT} } = \gamma^{\mathrm{DL}} \left( 1 - \frac{\tau_p}{\tau_c} \right) \log_2 \left( 1 + \mathrm{SINR}_k^{ \mathrm{MRT}} \right) \quad \textrm{[bit/symbol]},
\end{equation}
where the SINR, $\mathrm{SINR}_k^{ \mathrm{MRT}}$,  is
\begin{equation} \label{eq: SINR-MRT}
\frac{ M \sum\limits_{i =1}^{L} \rho_{i,k} \theta_{i,k} }{ M \sum\limits_{i=1}^{L} \sum\limits_{t \in \mathcal{P}_k \setminus \{k\} }  \rho_{i,t} \theta_{i,k}  +  \sum\limits_{i = 1 }^{L} \sum\limits_{ t=1 }^{K} \rho_{i,t} \beta_{i,k} + \sigma_{ \mathrm{DL}}^2 }.
\end{equation}
\begin{proof}
The proof is given in Appendix \ref{Appendix Corollary-MRT-Rate}.
\end{proof}
\end{corollary}
This corollary reveals the merits of MRT precoding for multi-cell Massive MIMO DL systems: The signal power increases proportionally to $M$ thanks to the array gain. The first term in the denominator is pilot contamination that increases proportionally to $M$ and makes the achievable rate saturated when $M \rightarrow \infty$ \cite{Jose2011b}. We also stress that a properly selected pilot reuse index set $\mathcal{P}_k$, for example the so-called pilot scheduling in \cite{Zhu2015a}, \cite{Xu2015a}, can significantly increase $\theta_{i,k}$ and thereby increase the SINR. In contrast, the regular interference is unaffected by the number of BS antennas. Finally, the non-coherent combination of received signals at user $k$ adds up the powers from multiple BSs and can give stronger signal gain than if only one BS serves the user.

\begin{corollary} \label{Corollary-ZF-Rate}
For Rayleigh fading  channels, if the BSs utilize ZF precoding, then the lower bound on the DL ergodic sum capacity in Theorem \ref{Theorem-Lower-Bound-Rate} is simplified to
\begin{equation}
R_k^{ \mathrm{ZF} } = \gamma^{\mathrm{DL}} \left( 1 - \frac{\tau_p}{\tau_c} \right) \log_2 \left( 1 + \mathrm{SINR}_k^{ \mathrm{ZF}} \right) \quad \textrm{[bit/symbol]},
\end{equation}
where the SINR, $\mathrm{SINR}_k^{ \mathrm{ZF}}$,  is
\begin{equation} \label{eq: SINR-ZF}
\frac{ (M-K) \sum\limits_{i =1}^{L} \rho_{i,k} \theta_{i,k} }{ (M-K) \sum\limits_{i =1}^{L} \sum\limits_{t \in \mathcal{P}_k \setminus \{k\} }  \rho_{i,t} \theta_{i,k} + \sum\limits_{ i =1 }^{L} \sum\limits_{ t=1 }^{K} \rho_{i,t} \left( \beta_{i,k} - \theta_{i,k} \right)+ \sigma_{ \mathrm{DL}}^2 }.
\end{equation}
\end{corollary}
\begin{proof}
The proof is given in Appendix \ref{Appendix Corollary-ZF-Rate}.
\end{proof}
The benefits of the array gain, BS non-coherent joint transmission, and pilot contamination effects shown by MRT are also inherited by ZF. The main distinction is that  MRT precoding only aims to maximize the signal-to-noise (SNR) ratio but does not pay attention to the multi-user interference. Meanwhile, ZF sacrifices some of the array gain to mitigate multi-user interference. The DL SE is limited by pilot contamination and the advantages of using mutually orthogonal pilot sequences are shown in Remark \ref{Remark1}.
\begin{remark} \label{Remark1}
When the number of BS antennas $M \rightarrow \infty$ and the number of users $K$ is fixed, the SINR values in \eqref{eq: SINR-MRT} for MRT and \eqref{eq: SINR-ZF} for ZF converge to $ \frac{ \sum_{i =1}^{L} \rho_{i,k} \theta_{i,k} }{ \sum_{i=1}^{L} \sum_{t \in \mathcal{P}_k \setminus \{k\} }  \rho_{i,t} \theta_{i,k} } $ meaning that the gain of adding more antennas diminishes. In contrast, if the users utilize mutually orthogonal pilot sequences, i.e., $\tau_p \geq K$, then adding up more BS antennas is always beneficial since the SINR value of user $k$ is given for MRT and ZF as 
\begin{equation} \label{eq: SINR-MRTOrthogonal}
\mathrm{SINR}_k^{\mathrm{MRT}} = 
\frac{ M \sum\limits_{i =1}^{L} \frac{\rho_{i,k} p_k \tau_p \beta_{i,k}^2 }{p_k \tau_p \beta_{i,k} + \sigma_{\mathrm{UL}}^2 } }{\sum\limits_{i = 1 }^{L} \sum\limits_{ t=1 }^{K} \rho_{i,t} \beta_{i,k} + \sigma_{ \mathrm{DL}}^2 },
\end{equation}
\begin{equation} \label{eq: SINR-ZFOrthogonal}
 \mathrm{SINR}_k^{\mathrm{ZF}} =  \frac{ (M -K) \sum\limits_{i =1}^{L} \frac{\rho_{i,k} p_k \tau_p \beta_{i,k}^2 }{p_k \tau_p \beta_{i,k} + \sigma_{\mathrm{UL}}^2 } }{\sum\limits_{i = 1 }^{L} \sum\limits_{ t=1 }^{K}  \frac{ \rho_{i,t} \beta_{i,k} \sigma_{\mathrm{UL}}^2}{p_k \tau_p \beta_{i,k} + \sigma_{\mathrm{UL}}^2} + \sigma_{ \mathrm{DL}}^2 }.
\end{equation}
\end{remark} 
 Note that for both MRT and ZF precoding, the DL ergodic SE not only depends on the channel estimation quality which can be improved by optimizing the pilot powers but also heavily depends on the power allocation at the BSs; that is, how the transmit powers $\rho_{i,t}$ are selected. In this paper, we only focus on the DL transmission, so \Cref{Section:Power-Optimization,Section:Optimal-Solution,Section:Max-Min-QoS} investigate different ways to jointly optimize the DL power allocation and user association with the predetermined pilot power.

\section{Downlink Transmit Power Optimization for Massive MIMO Systems} \label{Section:Power-Optimization}

The transmit power at BS $i$ depends on the traffic load over the coverage area and is limited by the peak radio frequency output power $P_{\mathrm{max},i}$, which defines the maximum power that can be utilized at each BS \cite{Bjornson2013e}. The transmit power $P_{ \mathrm{trans},  i}$ is computed as
\begin{equation} \label{eq: Transmit-Power}
P_{\mathrm{trans},i} = \mathbb{E} \{ \| \mathbf{x}_i \|^2 \} = \sum_{t=1 }^{K} \rho_{i,t}  \mathbb{E} \{ \| \mathbf{w}_{i,t} \|^2 \} = \sum_{t=1 }^{K} \rho_{i,t} .
\end{equation}
The transmit power consumption at BS $i$ that takes the power amplifier efficiency $ \Delta_i $ into account is modeled as
\begin{equation} \label{eq: Power-Consumption}
P_i = \Delta_i P_{\mathrm{trans},i}, \; 0 \leq P_{\mathrm{trans},i} \leq P_{\mathrm{max},i}.
\end{equation}
Here, $ \Delta_i $ depends on the BS technology \cite{EARTH_D23_short} and affects the power allocation and user association problems. Specifically, the values $\Delta_i$ may not be the same, for example, the BSs are equipped with the different hardware quality.

The main goal of a Massive MIMO network is to deliver a promised QoS to the users, while consuming as little power as possible. In this paper, we formulate it as a power minimization problem under user-specific SE constraints as
\begin{equation} \label{Optimization: General-Form}
\begin{aligned}
& \underset{ \{\rho_{i,t} \geq 0 \} }{\textrm{minimize}}
& & \sum_{ i=1 }^{L} P_i \\
& \textrm{subject to}
& &  R_k \geq \xi_k,\; \forall k \\
& &&  P_{\mathrm{trans},i} \leq P_{\mathrm{max}, i }, \; \forall i,\\
\end{aligned}
\end{equation}
where $\xi_k$ is the target QoS at user $k$. Plugging \eqref{eq: Sum-Rate-k}, \eqref{eq: Transmit-Power}, and \eqref{eq: Power-Consumption} into \eqref{Optimization: General-Form}, the optimization problem is converted to
\begin{equation} \label{Optimization: General-Form1}
\begin{aligned}
& \underset{ \{ \rho_{i,t} \geq 0 \} }{\textrm{minimize}}
& & \sum_{i =1 }^{L} \Delta_i \sum_{t=1}^{K} \rho_{i,t} \\
& \mbox{subject to}
& &  \mathrm{SINR}_{k}  \geq \hat{\xi}_k \;, \forall k\\
& &&  P_{\mathrm{trans},i} \leq P_{ \mathrm{max},i } \;, \forall i, \\
\end{aligned}
\end{equation}
where $\hat{\xi}_k = 2^{ \frac{\xi_k \tau_c }{ \gamma{\mathrm{DL}} (\tau_c - \tau_p)} } -1$ implies that the QoS targets are transformed into SINR targets. Owing to the universality of $\{\mathrm{SINR}_k\}$, \eqref{Optimization: General-Form1} is a general formulation for any selection of precoding scheme. We focus on MRT and ZF precoding since we have derived closed-form expressions for the corresponding SINRs. In these cases, the exact problem formulations are provided in Lemmas \ref{Optimization: MRT} and \ref{Optimization: ZF}.
\begin{lemma} \label{Optimization: MRT}
If the system utilizes MRT precoding, then the power minimization problem in \eqref{Optimization: General-Form1} is expressed as
\begin{equation} 
\begin{aligned}
& \underset{ \{ \rho_{i,t} \geq 0 \}}{\textrm{minimize}} &&
 \sum_{i=1}^{L} \Delta_i \sum\limits_{t=1}^{K} \rho_{i,t} \\
&\mbox{subject to} && \frac{ M \sum_{i =1}^{L} \rho_{i,k} \theta_{i,k} }{ M \sum\limits_{i=1}^{L} \sum\limits_{t \in \mathcal{P}_k \setminus \{k\} }  \rho_{i,t} \theta_{i,k}  +  \sum\limits_{i = 1 }^{L} \sum\limits_{ t=1 }^{K} \rho_{i,t} \beta_{i,k} + \sigma_{ \mathrm{DL}}^2 }  \\
& && \geq \hat{ \xi }_k, \; \forall k  \\
& && \sum_{t=1}^{K} \rho_{i,t} \leq P_{\mathrm{max},i}, \; \forall i. \\
\end{aligned}
\end{equation}
\end{lemma}

\begin{lemma} \label{Optimization: ZF}
If the system utilizes ZF precoding, then the power minimization problem in \eqref{Optimization: General-Form1} is expressed as
\begin{equation}
\begin{aligned}
& \underset{ \{ \rho_{i,t} \geq 0 \} }{\textrm{minimize}} && \sum_{i=1}^{L} \Delta_i \sum_{t=1}^{K} \rho_{i,t}   \\
&  \textrm{subject to} &&  \frac{ G \sum\limits_{i =1}^{L} \rho_{i,k} \theta_{i,k} }{ G \sum\limits_{i =1}^{L} \sum\limits_{\substack{t \in \\ \mathcal{P}_k \setminus \{k\} }}  \rho_{i,t} \theta_{i,k} + \sum\limits_{ i =1 }^{L} \sum\limits_{ t=1 }^{K} \rho_{i,t} \left( \beta_{i,k} - \theta_{i,k} \right)+ \sigma_{ \mathrm{DL}}^2 } \\
& && \geq \hat{ \xi }_k, \forall k \\
& && \sum_{t=1}^{K} \rho_{i,t} \leq P_{\mathrm{max},i}, \forall i,\\
\end{aligned}
\end{equation}
where $G= M -K$.
\end{lemma}

The optimal power allocation and user association are obtained by solving these problems. At the optimal solution, each user $t$ in the network is associated with the subset of BSs that is determined by the non-zero values $\rho_{i,t}, \forall i, t$. The BS-user association problem is thus solved implicitly. There are fundamental differences between our problem formulation and the previous ones that appeared in \cite{Li2015,Stridh2006a, Rashid1998b} for conventional MIMO systems with a few antennas at the BSs. The main distinction is that these previous works consider short-term QoS constraints that depend on the current fading realizations, while we consider long-term QoS constraints that do not depend on instantaneous fading realizations thanks to channel hardening and favorable properties in Massive MIMO. In addition, our proposed approach is more practically appealing since the power allocation and BS-user association can be solved over a longer time and frequency horizons and since we do not try to combat small-scale and frequency-selective fading by the power control.
\vspace*{-0.5cm}
\section{Optimal Power Allocation and User Association by linear programming} \label{Section:Optimal-Solution}
This section provides a unified mechanism to obtain the optimal solution to the total power minimization problem for both MRT and ZF precoding. The BS-user association principle is also discussed by utilizing Lagrange duality theory. 
\vspace*{-1cm}
\subsection{Optimal Solution with Linear Programming}
We now show how to obtain optimal solutions for the problems stated in Lemmas \ref{Optimization: MRT} and  \ref{Optimization: ZF}.  Let us denote the power control vector of an arbitrary user $t$ by $\pmb{\rho}_t = [\rho_{1, t}, \ldots, \rho_{L,t}]^T \in \mathbb{C}^{L}$, where its entries satisfy $\rho_{i,t} \geq 0$ meaning that $ \pmb{\rho}_t  \succeq 0$. We also denote $\pmb{\Delta} = [\Delta_1 \ldots \Delta_L]^T \in \mathbb{C}^{L}$ and $\pmb{\epsilon}_i \in \mathbb{C}^L$ has all zero entries but the $i$th one is 1. The optimal power allocation is obtained by the following theorem.

\begin{theorem} \label{Theorem: Linear-Solution}
The optimal solution to the total transmit power minimization problem in \eqref{Optimization: General-Form1} for MRT or ZF precoding is obtained by solving the linear program
\begin{equation} \label{eq: Linear-Solution-CVX}
\begin{aligned}
& \underset{ \{ \pmb{\rho}_t \succeq 0 \}  }{\textrm{minimize}}
& & \sum_{t=1}^{K} \pmb{\Delta}^T \pmb{\rho}_t \\
& \textrm{subject to}
& &   \sum_{t \in \mathcal{P}_k \setminus \{k \} } \pmb{\theta}_k^T \pmb{\rho}_t + \sum_{t=1}^{K} \mathbf{c}_k^T \pmb{\rho}_t -  \mathbf{b}_k^T \pmb{\rho}_k + \sigma_{ \mathrm{DL} }^2 \leq 0, \; \forall k  \\
& && \sum_{t=1}^{K} \pmb{\epsilon}_i^T \pmb{\rho}_{t} \leq P_{\mathrm{max},i}, \; \forall i.\\
\end{aligned}
\end{equation}
Here, the vectors $\pmb{\theta}_{k}, \mathbf{c}_{k}, $ and $\mathbf{b}_{k} $ depend on the precoding scheme. MRT precoding gives
\begin{equation*}
\begin{split}
\pmb{\theta}_k &=\left[ M \theta_{1,k}, \ldots, M \theta_{L,k} \right]^T \\
\mathbf{c}_k &= \left[ \beta_{1,k}, \ldots, \beta_{L,k}  \right]^T,\\
\mathbf{b}_k &= \left[ M \theta_{1,k} / \hat{\xi}_k, \ldots, M \theta_{L,k} / \hat{\xi}_k \right]^T,
\end{split}
\end{equation*}
while ZF precoding obtains
\begin{equation*}
\begin{split}
\pmb{\theta}_k &=\left[ (M -K) \theta_{1,k}, \ldots, (M -K) \theta_{L,k} \right]^T \\
\mathbf{c}_k &= \left[ \beta_{1,k} - \theta_{1,k}, \ldots, \beta_{L,k} - \theta_{L,k} \right]^T,\\
\mathbf{b}_k &= \left[ (M -K ) \theta_{1,k} / \hat{\xi}_k , \ldots, (M -K )  \theta_{L,k} / \hat{\xi}_k \right]^T.
\end{split}
\end{equation*} 
\end{theorem}
\begin{proof}
The problem in \eqref{eq: Linear-Solution-CVX} is obtained from Lemmas~\ref{Optimization: MRT} and ~\ref{Optimization: ZF} after some algebra. We note that the objective function is a linear combination of $\pmb{\rho}_t$, for $t=1, \ldots, K$. Moreover, the constraint functions are affine functions of power variables. Thus the optimization problem \eqref{eq: Linear-Solution-CVX} is a linear program.
\end{proof}

The merits of Theorem \ref{Theorem: Linear-Solution} are twofold:  It indicates that the total transmit power minimization problem for a multi-cell Massive MIMO system with non-coherent joint transmission is linear and thus can be solved to global optimality in polynomial time, for example, using general-purpose implementations of interior-point methods such as CVX \cite{cvx2015}. \footnote{The linear program in \eqref{eq: Linear-Solution-CVX} is only obtained for non-coherent joint transmission. For the corresponding system that deploys another CoMP technique called coherent joint transmission, the total transmit power optimization with Rayleigh fading channels and MRT or ZF precoding is a second-order cone program (see Appendix \ref{Appendix: Coherent Joint Transmission}). This problem is considered in Section \ref{Section:Numerical-Results} for comparison reasons.} In addition, the solution provides the optimal BS-user association in the system. We further study it via Lagrange duality theory in the next subsection.

\subsection{BS-User Association Principle}
To shed light on the optimal BS-user association provided by the solution in Theorem \ref{Theorem: Linear-Solution}, we analyze the problem utilizing Lagrange duality theory. The Lagrangian of \eqref{eq: Linear-Solution-CVX} is
\begin{equation} \label{eq: Langrangian}
\begin{split}
& \mathcal{L} ( \pmb{ \rho}_t, \lambda_k, \mu_i ) =  \sum_{ t= 1}^{K} \pmb{\Delta}^T \pmb{\rho}_t   \\
& \; \; \; + \sum_{k=1}^{K} \lambda_k \left( \sum_{t \in \mathcal{P}_k \setminus \{k \} } \pmb{\theta}_k^T \pmb{\rho}_t + \sum_{t=1}^{K} \mathbf{c}_k^T \pmb{\rho}_t - \mathbf{b}_k^T \pmb{\rho}_k + \sigma_{\mathrm{DL}}^2 \right)  \\
&\; \; \;+ \sum_{i=1}^{L}\mu_i \left( \sum_{t=1}^{K} \pmb{\epsilon}_i^T \pmb{\rho}_t - P_{\mathrm{max},i} \right),
\end{split}
\end{equation}
where the non-negative Lagrange multipliers $\lambda_k$ and $\mu_i$ are associated with the $k$th QoS constraint and the transmit power constraint at BS $i$, respectively. The corresponding Lagrange dual function of \eqref{eq: Langrangian} is formulated as
\begin{equation} \label{eq: Langrangian-Duality}
\begin{split}
& \mathcal{G}\left(\lambda_k, \mu_i \right) = \underset{  \{ \pmb{\rho}_t \} }{ \inf} \; \mathcal{L} \left( \pmb{\rho}_t, \lambda_k, \mu_i \right) \\
&= \sum_{k=1}^{K} \lambda_k \sigma_{ \mathrm{DL} }^2 - \sum_{i=1}^{L} \mu_i P_{\mathrm{max},i} + \underset{ \{ \pmb{\rho}_t \} }{ \inf } \; \sum_{t=1}^{K} \mathbf{a}_t^T \pmb{\rho}_t,
\end{split}
\end{equation}
where $\mathbf{a}_t^T = \pmb{\Delta}^T + \sum_{k=1}^K \lambda_k \pmb{\theta}_k^T \mathbbm{1}_k(t)+ \sum_{k=1}^{K} \lambda_k \mathbf{c}_k^T - \lambda_t \mathbf{b}_t^T + \sum_{i=1}^{L} \mu_i \pmb{ \epsilon }_i^T $ and  the indicator function $\mathbbm{1}_k (t)$ is defined as
\begin{equation}
\mathbbm{1}_k (t) = \begin{cases}  
 0, &  t \notin \mathcal{P}_k \setminus \{k\} , \\
1 , & t \in \mathcal{P}_k \setminus \{k\}. 
\end{cases}
\end{equation}
It is straightforward to show that $\mathcal{G}\left(\lambda_k, \mu_i \right)$ is bounded from below (i.e, $\mathcal{G}\left(\lambda_k, \mu_i \right) \neq - \infty$) if and only if $\mathbf{a}_t \succeq 0$, for $t = 1, \ldots, K$. Therefore, the Lagrange dual problem to \eqref{eq: Linear-Solution-CVX} is
\begin{equation} \label{Optimization: Dual-Problem}
\begin{aligned}
& \underset{ \{  \lambda_k , \mu_i \} }{\textrm{maximize}}
& & \sum_{k=1}^{K} \lambda_k \sigma_{ \mathrm{DL} }^2 - \sum_{i=1}^{L} \mu_i P_{\mathrm{max},i} \\
& \textrm{subject to}
& & \mathbf{a}_t \succeq 0, \; \forall t. \\
\end{aligned}
\end{equation}
From this dual problem, we obtain the following main result that gives the set of BSs serving an 
arbitrary user $t$. 
\begin{theorem} \label{Theorem-BS-Association}
Let $\{ \check{\lambda}_k, \check{\mu}_i \}$ denote the optimal Lagrange multipliers. User $t$ is served only by the subset of BSs with indices in the set $\mathcal{S}_t $ defined as
\begin{equation} \label{eq: BS-Association}
\underset{i}{ \mathrm{argmin}} \left( \Delta_i  + \sum\limits_{k=1}^K \check{\lambda}_k \theta_{i,k} \mathbbm{1}_k (t) +  \sum\limits_{k=1}^{K} \check{\lambda}_{k} c_{i,k} + \sum\limits_{i=1}^{L} \check{\mu}_{i} \right)\frac{ 1 }{b_{i,t}},
\end{equation}
where the parameters $c_{i,k}$ and $b_{i,t}$ are selected by the linear precoding scheme:
\vspace{0.2cm}
\begin{center}
\begin{tabular}{| c|  c |c |}
\hline
 Precoding scheme & $c_{i,k}$ & $b_{i,t}$ \\ 
 \hline
 MRT & $\beta_{i,k}$ & $ M \theta_{i,t} / \hat{\xi}_t  $ \\  
 \hline
 ZF & $ \beta_{i,k} - \theta_{i,k} $ & $ (M-K) \theta_{i,t} / \hat{\xi}_t  $ \\  
 \hline  
\end{tabular}
\end{center}
\vspace{0.2cm}
The optimal BS association for user $t$ is further specified as one of the following two cases:
\begin{itemize}
\item It is served by one BS if the set $\mathcal{S}_t$ in \eqref{eq: BS-Association} only contains one index.
\item It is served by multiple of BSs if the set $\mathcal{S}_t$ in \eqref{eq: BS-Association} contains several indices.
\end{itemize}
\end{theorem}
\begin{proof}
The proof is given in Appendix \ref{Appendix BS-Association}.
\end{proof}
 
The expression in \eqref{eq: BS-Association} explicitly shows that the optimal BS-user association is affected by many factors such as interference between BSs, noise intensity level, power allocation, large-scale fading, channel estimation quality, pilot contamination, and QoS constraints. There is no simple user association rule since the function depends on the Lagrange multipliers, but we can be sure that max-SNR association is not always optimal. We will later show numerically that for Rayleigh fading channels and MRT or ZF, each user is usually served by only one BS at the optimal point.

\section{Max-min QoS Optimization} \label{Section:Max-Min-QoS}

This section is inspired by the fact that there is not always a feasible solution to the power minimization problem with fixed QoS constraints in \eqref{eq: Linear-Solution-CVX}. The reason is the trade-off between the target QoS constraints and the propagation environments. The path loss is one critical factor, while limited power for pilot sequences leads to that channel estimation error always exists. Thus, for a certain network, it is not easy to select the target QoS values. In order to find appropriate QoS targets, we consider a method to optimize the QoS constraints along with the power allocation.

Fairness is an important consideration when designing wireless communication systems to provide uniformly great service for everyone \cite{Yang2014a}. The vision is to provide a good target QoS to all users by maximizing the lowest QoS value, possibly with some user specific weighting. For this purpose, we consider the optimization problem
\begin{equation} \label{Problem: Max-Min-QoS}
\begin{aligned}
& \underset{ \{ \rho_{i,t} \geq 0 \} }{\textrm{maximize}} \; \underset{k}{\textrm{min}}
& &  R_k/w_k \\
& \textrm{subject to}
& & P_{ \mathrm{trans},i} \leq P_{\mathrm{max},i} \;, \forall i  ,\\
\end{aligned}
\end{equation}
where $w_k > 0$ is the weight for user $k$. The weights can be assigned based on for example information about the propagation, interference situation or user priorities. If there is no such explicit priorities, they may be set to $1$. To solve \eqref{Problem: Max-Min-QoS}, it is converted to the epigraph form \cite{Boyd2004a}
\begin{equation} \label{eq: WeightSpecific2}
\begin{aligned}
& \underset{ \{ \rho_{i,t} \geq 0 \}, \xi}{\textrm{maximize}}
& & \xi \\
& \textrm{subject to}
& & R_{k}/ w_k \geq \xi \;, \forall k \\
& & & P_{ \mathrm{trans},i} \leq P_{\mathrm{max},i} \;, \forall i,\\
\end{aligned}
\end{equation}
where $\xi$ is the minimum QoS parameter for the users that we aim to maximize. Plugging \eqref{eq: Sum-Rate-k} and  \eqref{eq: Transmit-Power} into \eqref{eq: WeightSpecific2}, we obtain

\begin{equation} \label{eq: WeightSpecific3}
\begin{aligned}
& \underset{ \{ \rho_{i,t} \geq 0 \}, \xi}{\textrm{maximize}}
& & \xi \\
& \textrm{subject to}
& &  \mathrm{SINR}_k \geq 2^{ \xi w_k /( \gamma^{ \mathrm{DL} } \left(1 -  \tau_p/\tau_c  \right) ) } -1  \;, \forall k \\
& & & \sum_{t=1}^{K} \rho_{i,t} \leq P_{\mathrm{max},i}, \; \forall i.\\
\end{aligned}
\end{equation}
We can solve \eqref{eq: WeightSpecific3} for a fixed $\xi$ as a linear program, using Theorem \ref{Theorem: Linear-Solution} with $\xi_k = \xi w_k$. Since the QoS constraints are increasing functions of $\xi$, the solution to the max-min QoS optimization problem is obtained by doing a line search over $\xi$ to get the maximal feasible value. Hence, this is a quasi-linear program. As a result, we further apply Lemma 2.9 and Theorem 2.10 in \cite{Bjornson2013d} to obtain the solution as follows.
\begin{theorem} \label{Theorem-Bisection}
The optimum to \eqref{eq: WeightSpecific3} is obtained by checking the feasibility of \eqref{eq: Linear-Solution-CVX} over an SE search range $\mathcal{R} = [0 , \xi_0^{\mathrm{upper}} ]$, where $\xi_0^{\mathrm{upper}}$ is selected to make \eqref{eq: Linear-Solution-CVX} infeasible.
\end{theorem}
\begin{corollary} \label{Corollary-Bisection}
If the system deploys MRT or ZF precoding, then $\xi_0^{\mathrm{upper}}$ can be selected as

\begin{equation}
\xi_0^{\mathrm{upper}} = \gamma^{\mathrm{DL}} \left( 1 - \frac{\tau_p}{\tau_c} \right) \theta.
\end{equation}
The parameter $\theta$ depends on the precoding scheme:
\vspace{0.2cm}
\begin{center}
\begin{tabular}{| c|  c |}
\hline
  &  $ \theta $ \\ 
 \hline
 MRT & $\underset{ k }{ \min } \; \frac{1}{w_k} \log_2 (1 + M ) $ \\  
 \hline
 ZF & $ \underset{ k }{ \min } \; \frac{1}{w_k} \log_2 \left(1 + ( M - K) \frac{ p_k \tau_p }{ \sigma_{ \mathrm{UL} }^2} \sum\limits_{i =1}^{L} \beta_{i,k} \right)$\\  
 \hline  
\end{tabular}
\end{center}
\vspace{0.2cm}
\end{corollary}
\begin{proof}
The proof is given in Appendix \ref{Appendix Bisection}.
\end{proof}
\begin{algorithm}[h]
\caption{Max-min QoS based on the bisection method}
\textbf{Result:} Solve optimization in \eqref{Problem: Max-Min-QoS}. 
\\ \textbf{Input:}  Initial upper bound $\xi_0^{\mathrm{upper}}$, and line-search accuracy $\delta$;
\begin{algorithmic}
\State Set  $\xi^{\mathrm{lower}}= 0$; $\xi^{\mathrm{upper}}= \xi_0^{\mathrm{upper}}$;
\While {$\xi^{\mathrm{upper}} - \xi^{\mathrm{lower}} > \delta$}
\State Set  $\xi^{\mathrm{candidate}} = (\xi^{\mathrm{upper}}+\xi^{\mathrm{lower}})/2$;
\If { \eqref{eq: Linear-Solution-CVX} is infeasible for $ \xi_k = w_k \xi^{\mathrm{candidate}}, \forall k,$ } \do
\\
 \State Set $\xi^{\mathrm{upper}} = \xi^{\mathrm{candidate}}$;
\Else

\State Set $ \{ \pmb{\rho}_{k}^{\mathrm{lower}} \}$ as the solution to \eqref{eq: Linear-Solution-CVX};

\State Set $\xi^{\mathrm{lower}} = \xi^{\mathrm{candidate}}$ ;
\EndIf
\EndWhile
\State Set $\xi_{\mathrm{final}}^{\mathrm{lower}} = \xi^{\mathrm{lower}}$ and $\xi_{\mathrm{final}}^{\mathrm{upper}} = \xi^{\mathrm{upper}}$;
\end{algorithmic}
\textbf{Output:} Final interval $[\xi_{\mathrm{final}}^{\mathrm{lower}}, \xi_{\mathrm{final}}^{\mathrm{upper}}]$ and $ \{ \tilde{\pmb{\rho}}_{k} \}= \{ \pmb{\rho}_{k}^{\mathrm{lower}} \}$;
\label{Algorithm: Bisection}
\end{algorithm}

From Theorem \ref{Theorem-Bisection}, the problem \eqref{eq: WeightSpecific3} is solved in an iterative manner. By iteratively reducing the search range and solving the problem \eqref{eq: Linear-Solution-CVX}, the maximum QoS level and optimal BS-user association can be obtained. One such line search procedure is the well-known bisection method \cite{Bjornson2013d, Boyd2004a}. At each iteration, the feasibility of \eqref{eq: Linear-Solution-CVX} is verified for a value $\xi^{\mathrm{candidate}} \in \mathcal{R}$, that is defined as the middle point of the current search range. If \eqref{eq: Linear-Solution-CVX} is feasible, then its solution $\{\pmb{\rho}_k^{\mathrm{lower}} \}$ is assigned to as the current power allocation. Otherwise, if the problem is infeasible, then a new upper bound is set up. The search range reduces by half after each iteration, since either its lower or upper bound is assigned to $\xi^{\mathrm{candidate}}$. The algorithm is terminated when the gap between these bounds is smaller than a line-search accuracy value $\delta$. The proposed max-min QoS optimization is summarized in Algorithm \ref{Algorithm: Bisection}. 

We stress that the bisection method can efficiently find the solution to quasi-linear programs such as \eqref{eq: WeightSpecific3}. The main cost for each iteration is to solve the linear program \eqref{eq: Linear-Solution-CVX} that includes $K L$ variables and $2K$ constraints and as such it has the complexity $\mathcal{O} ( K^3 L^3 )$ \cite{Boyd2004a}. It is important to note that the computational complexity does not depend on the number of BS antennas.  Moreover, the number of iterations needed for the bisection method is $ \lceil \log_2 (\xi_0^{\textrm{upper}}/\delta) \rceil$ that is directly proportional to the logarithm of the initial value $\xi_0^{\mathrm{upper} }$, where $\lceil \cdot \rceil$ is the ceiling function. Thus a proper selection for $\xi_0^{\mathrm{upper} }$ such as in Corollary \ref{Corollary-Bisection} will reduce the total cost. In summary, the polynomial complexity of Algorithm \ref{Algorithm: Bisection} is $\mathcal{O} \left(    \left \lceil \log_2 \left( \frac{\xi_0^{\textrm{upper}}}{ \delta } \right) \right \rceil K^3 L^3 \right).$

\section{Numerical Results} \label{Section:Numerical-Results}

In this section, the analytical contributions from the previous sections are evaluated by simulation results for a multi-cell Massive MIMO system. Our system comprises $4$ BSs and $K$ users, as shown in Fig.~\ref{fig:MassiveMIMOSystem-Layout}, where $(x,y)$ represent location in a Cartesian coordinate system. The symmetric BS deployment makes it easy to visualize the optimal user association rule. The users are uniformly and randomly distributed over the joint coverage of the BSs but no user is closer to the BSs than $100$ m to avoid overly large SNRs at cell-center users \cite{Marzetta2010a}. For the max-min QoS algorithm, the user specific weights are set to $w_k =1$, $ \forall k$, to make it easy to interpret the results. Since the joint power allocation and user association obtains the optimal subset of BSs that serves each user, we denote it by ``Optimal" in the figures. For comparison, we consider a suboptimal method, in which each user is associated with only one BS by selecting the strongest signal on the average (i.e., the max-SNR value). \footnote{For comparison purposes, the best benchmark is the method that also performs the optimal association but with service from only one BS. However, it is a combinatorial problem followed by the excessive computational complexity. Furthermore, the numerical results verify that the max-SNR association is a good benchmark for comparison since the performance is very close to the optimal association.} The performance is averaged over different random user locations.
\begin{figure}[t]
    \centering
    \includegraphics[width=2.3in]{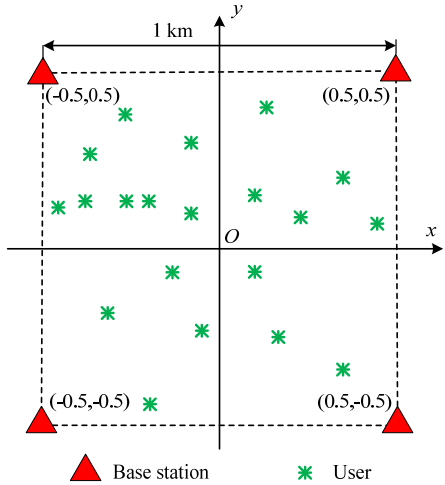}
    \caption{ The multi-cell massive MIMO system considered in the simulations: The BS locations are fixed at the corners of a square, while $K$ users are randomly distributed over the joint coverage of the BSs.}
    \label{fig:MassiveMIMOSystem-Layout}
\end{figure}
The peak radio frequency DL output power is $40$ W. The system bandwidth is $20$ MHz and the coherence interval is of $200$ symbols. We set the power amplifier efficiency to $1$ since it does not affect the optimization when all BSs have the same value. The users send orthogonal pilot sequences whose length equals the number of users and each user has a pilot symbol power of $200$ mW. \footnote{ In most cases in practice, appropriate non-universal pilot reuse renders pilot contamination negligible. Hence, we only consider the case of orthogonal pilot sequences in this section. We also assume $\tau_p = K$.}
 Because we focus on the DL transmission, the DL fraction is $\gamma^{\mathrm{DL}}=1$. The large scale fading coefficients are modeled similarly to the 3GPP LTE standard \cite{LTE2010b, Bjornson2013e}. Specifically, the shadow fading $z_{l,k}$ is generated from a log-normal Gaussian distribution with standard deviation $7$ dB. The path loss at distance $d$ km is $148.1 + 37.6 \log_{10} d$. Thus, the large-scale fading $\beta_{l,k}$ is computed by $\beta_{l,k}= - 148.1 - 37.6 \log_{10} d + z_{l,k}$ dB.  With the noise figure of $5$ dB, the noise variance for both the UL and DL is $-96$ dBm.
  \begin{figure}[t]
  		\centering
  		\includegraphics[trim=2.5cm 0.9cm 2cm 1.5cm, clip=true, width=3in]{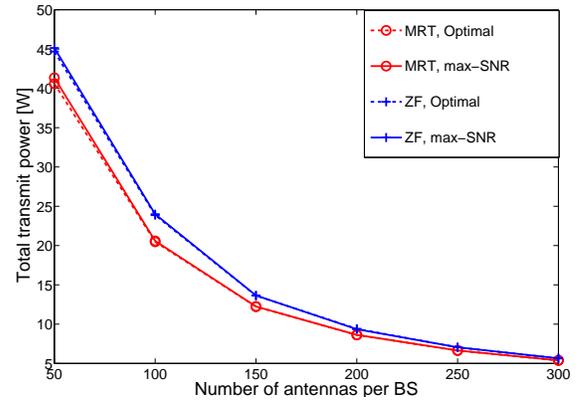}
  		\caption{ The total transmit power $( \sum_{i=1}^{L} P_i )$ versus the number of BS antennas with QoS of $1$ bit/symbol and $K=20$.}
  		\label{Fig-PowervsAntenna}
  		\vspace*{-0.5cm}
\end{figure}
We show the total transmit power $( \sum_{i=1}^{L} P_i)$ as a function of the number of antennas per BS in Fig.~\ref{Fig-PowervsAntenna} for a Massive MIMO system with $20$ users. For fair comparison, the results are averaged over the solutions that make both the association schemes feasible. Experimental results reveal a superior reduction of the total transmit power compared to the peak value, say $160$ W, in current wireless networks. Therefore, Massive MIMO can bring great transmit power reduction by itself. A system equipped with few BS antennas consumes much more transmit power to provide the same target QoS level compared to the corresponding one with a large BS antenna number. The $40-45$ W that are required with $50$ BS antennas reduces dramatically to $5$ W with $300$ BS antennas. This is due to the array gain from coherent precoding. In addition, the gap between MRT and ZF is shortened by the number of BS antennas, since interference is mitigated more efficiently \cite{Ngo2014a, Bjornson2016b}. From the experimental results, we notice that the simple max-SNR association is close to optimal in these cases.

Fig.~\ref{Fig-PowervsQoS} plots the total transmit power to obtain various target QoS levels at the $20$ users. As discussed in Section \ref{Achievable-Spectral-Efficiency}, MRT precoding works well in the low QoS regime where noise dominates the system performance, while ZF precoding consumes less power when higher QoS is required. In the low QoS regime, ZF and MRT precoding demand roughly the same transmit power. For instance, with the optimal BS-user association and QoS $= 1$ bit/symbol, the system requires the total transmit power of $8.88$ W and $9.60$ W for MRT and ZF precoding, respectively. In contrast, at a high target QoS level such as $2.5$ bit/symbol, by deloying ZF rather than MRT, the system saves transmit power up to $2.39$ W. Similar trends are observed for the max-SNR association. Because the numerical results manifest superior power reduction in comparison to the small-scale MIMO systems, Massive MIMO systems are well-suited for reducing transmit power in $5$G networks.

\begin{figure}[t]
	\centering
	\includegraphics[trim=2.5cm 0.9cm 2cm 1.5cm, clip=true, width=3in]{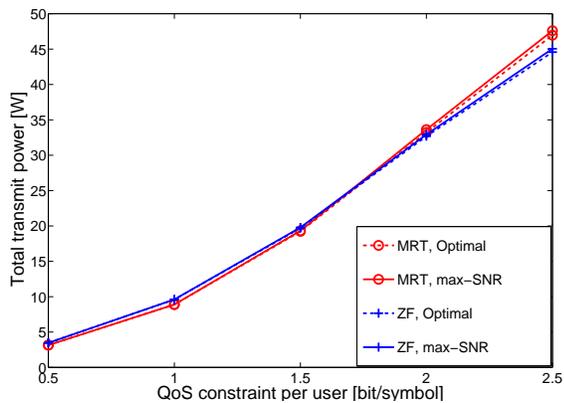}
	\caption{ The total transmit power $( \sum_{i=1}^{L} P_i )$ versus the target QoS for $M= 200 , K = 20$.}
	\label{Fig-PowervsQoS}
	\vspace*{-0.5cm}
\end{figure}
 \begin{figure}[t]
  \centering
  \includegraphics[trim=2.0cm 0.9cm 2cm 1.5cm, clip=true, width=3in]{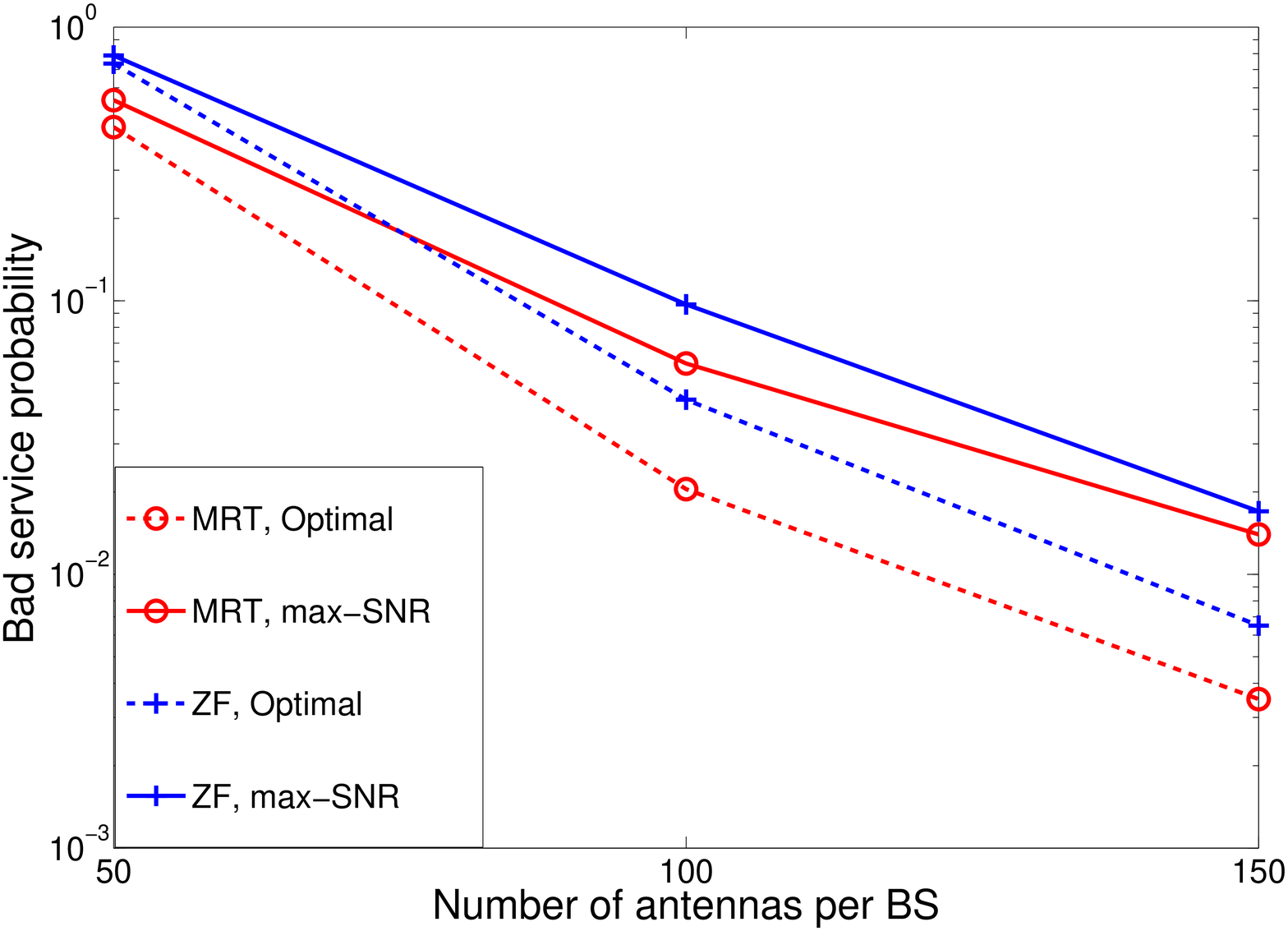}
  \caption{ The bad service probability versus the number of BS antennas with QoS of $1$ bit/symbol and $K=20$.}
  \label{Fig-InfeasibleAntennas}
  \vspace*{-0.5cm}
  \end{figure}

While optimal user association and max-SNR association give similar results in the previous figures, we stress that these only considered scenarios when both schemes gave feasible results. The main difference is that sometimes only the former can satisfy the QoS constraints. Fig.~\ref{Fig-InfeasibleAntennas} and Fig.~\ref{Fig-InfeasibleQoS} demonstrate the ``bad service probability" defined as the fraction of random user locations and shadow fading realizations in which the power minimization problem is infeasible. Note that these figures just display some ranges of the BS antennas or the QoS constraints where the differences between the user associations are particularly large. Intuitively, the optimal user association is more robust to environment variations than the max-SNR association, since the non-coherent joint transmission can help to resolve the infeasibility. In addition, the two figures also verify the difficulties in providing the high QoS. Specifically, a very high infeasibility up to about $80\%$ is observed when the BSs have a small number of antennas or the users demand high QoS levels. This is a key motivation to consider the max-min
QoS optimization problem instead, because it provides feasible solutions for any user locations and channel realizations.

Fig.~\ref{Fig-CDFofQoS} shows the cumulative distribution function (CDF) of the max-min optimized QoS level, where the randomness is due to shadow fading and different user locations. We consider $150$ BS antennas for ZF precoding or $300$ BS antennas for MRT precoding to avoid overlapping curves. The optimal user association gives consistently better QoS than the max-SNR association. The system model equipped with $300$ antennas per BS can provide SE greater than $2$ bit/symbol for every user terminal in its coverage area with high probability. The QoS can even reach up to $4$ bit/symbol. Moreover, the optimal association gains up to $22\%$ compared with the max-SNR association at $95\%$-likely max-min QoS.
  \begin{figure}[t]
  	\centering
  	\includegraphics[trim=2.0cm 0.9cm 2cm 1.5cm, clip=true, width=3 in]{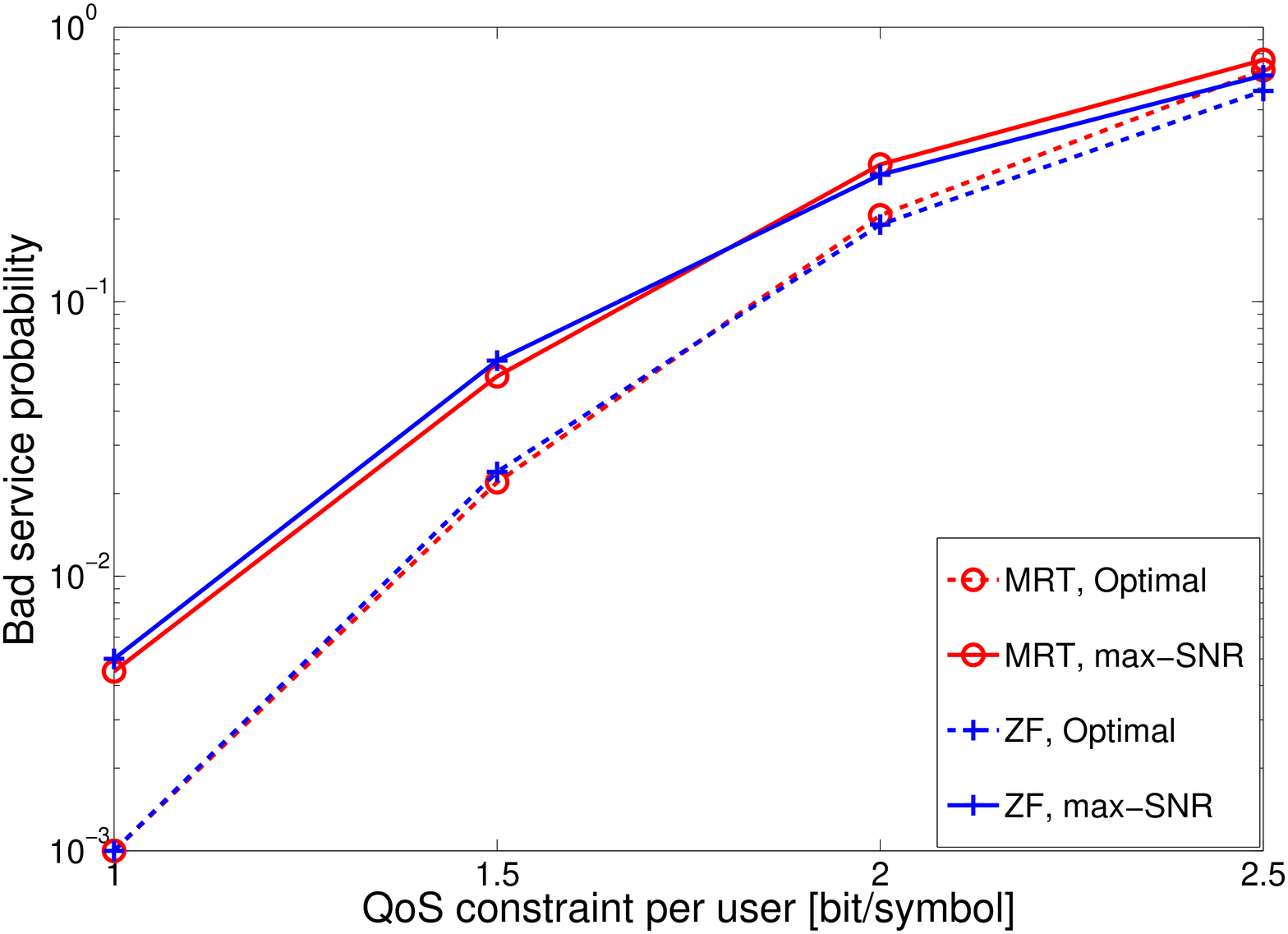}
  	\caption{The bad service probability versus the QoS constraint per user with $M=200, K=20$.}
  	\label{Fig-InfeasibleQoS}
  	\vspace*{-0.5cm}
  \end{figure}
\begin{figure}[t]
 \centering
 \includegraphics[trim=2.0cm 0.9cm 2cm 1.5cm, clip=true, width=3in]{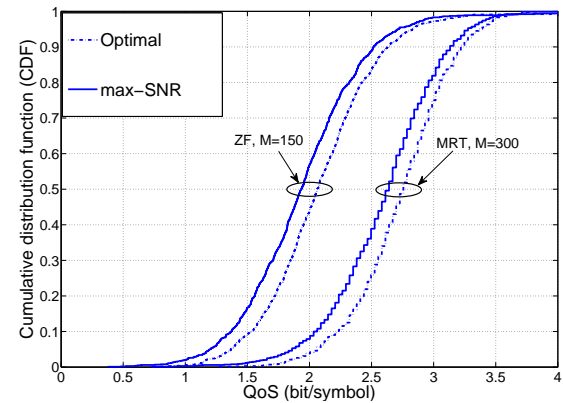}
 \caption{ The cumulative distribution function (CDF) of the max-min QoS optimization with $K=20$.}
 \label{Fig-CDFofQoS}
 \vspace*{-0.5cm}
  \end{figure}
 \begin{figure}[t]
       \centering
       \includegraphics[trim=2.0cm 0.9cm 2cm 1.5cm, clip=true, width=3 in]{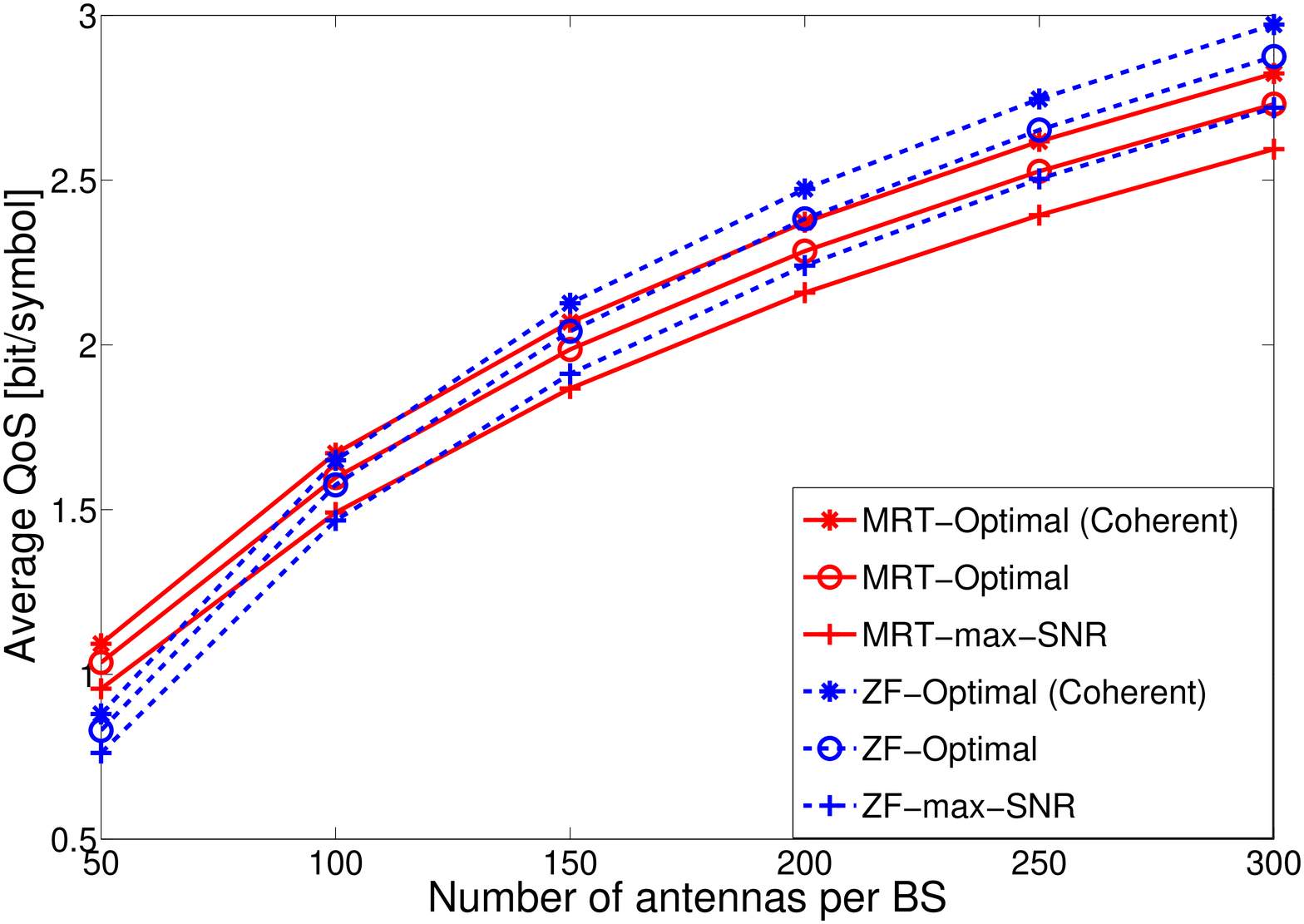}
       \caption{The average max-min QoS level versus the number of BS antennas with $K= 20$.}
       \label{Fig-AverageMaxMinQoS} 
\vspace*{-0.5cm}
\end{figure}

\begin{figure}[t]

      \centering
       \includegraphics[trim=2.0cm 0.9cm 2cm 1.5cm, clip=true, width=3 in]{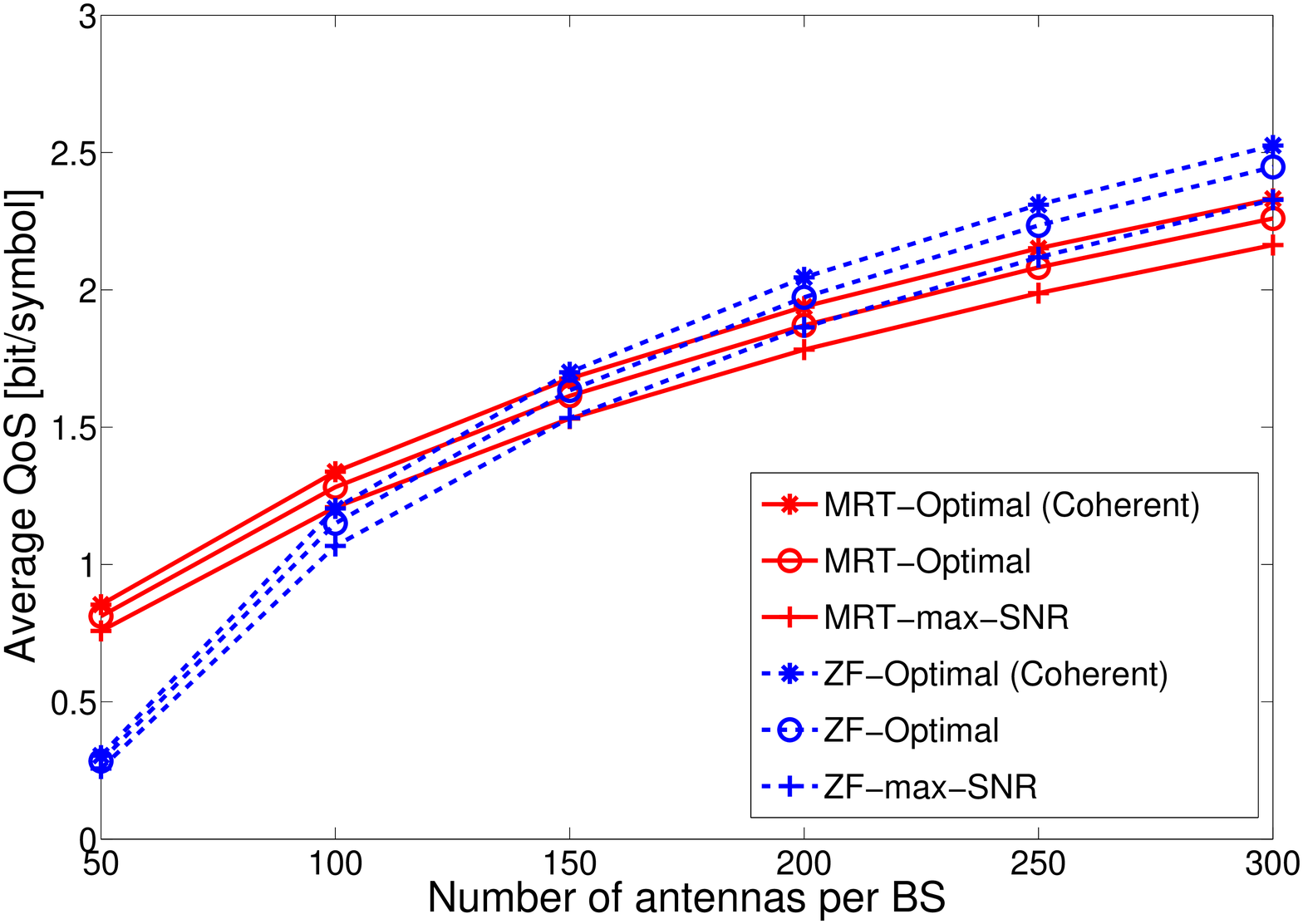}
       \caption{The average max-min QoS level versus the number of BS antennas with $K= 40$.}
      \label{Fig-AverageMaxMinQoS40} 
      \vspace*{-0.5cm}
\end{figure}  

\begin{figure}[t]
\centering
\includegraphics[trim=2.0cm 0.9cm 2cm 1.5cm, clip=true, width=3 in]{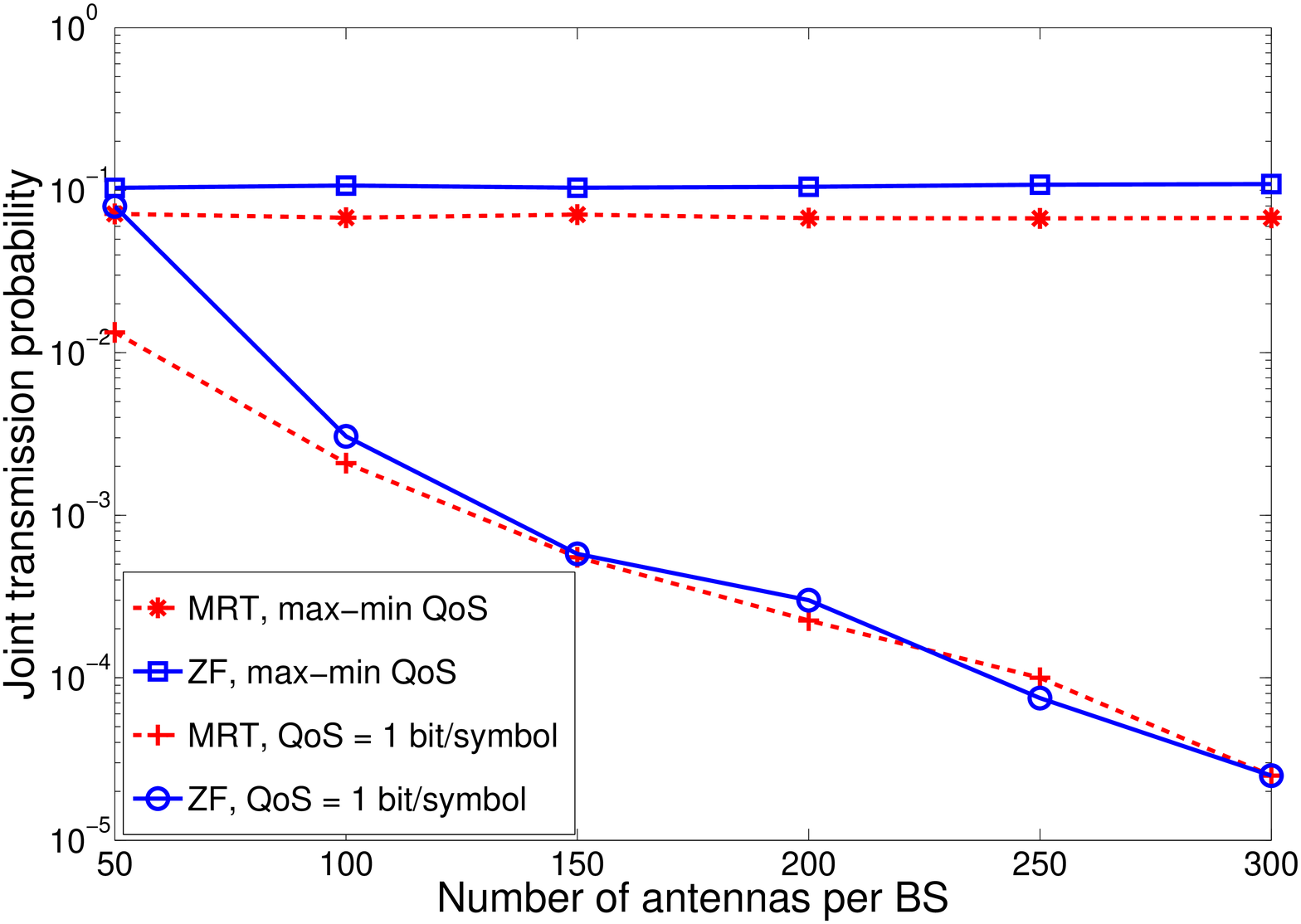}
\caption{ The joint transmission probability versus the number of antennas per BS with $K=20$.}
\label{Fig-BSAssociationProbability}     
\vspace*{-0.5cm} 
\end{figure}

The average max-min QoS levels that the system can provide to the all users is illustrated in Fig.~\ref{Fig-AverageMaxMinQoS} for $20$ users. The optimal BS-user association provides up to $11\%$ higher QoS than the max-SNR association. For completeness, we also provide the average max-min QoS levels when the system deploys the DL coherent joint transmission denoted as ``Optimal (Coherent)" in the figure. The procedures to obtain closed-form expressions as well as the optimization problems for the DL coherent joint transmission are briefly presented in Appendix \ref{Appendix: Coherent Joint Transmission}. On average, this technique can bring a gain up to $5\%$ compared to ``Optimal" but it is more complicated to implement as discussed in Section \ref{Downlink Data Transmission Model}. By deploying massive antennas at the BSs, the numerical results manifest the competitiveness of  the max-SNR association versus the ``Optimal" ones. The reason is that the multiple BS cooperation increases not only the array gain (in the numerator) but unfortunately also mutual interference (in the denominator) of the SINR expressions as shown in Corollaries \ref{Corollary-MRT-Rate} and \ref{Corollary-ZF-Rate} for non-coherent joint transmission or in \eqref{eq: Rate_Coherent} for coherent joint transmission. It is only a few users that gain from non-coherent joint transmission and the added benefit from coherent joint transmission is also small.

Fig.~\ref{Fig-AverageMaxMinQoS40} considers the same setup as Fig.~\ref{Fig-AverageMaxMinQoS} but with $40$ users. Here, the max-min QoS reduces due to more interference, while the gain from joint transmission is still small. When the number of antennas per BS is not significantly larger than the number users, MRT outperforms ZF because ZF sacrifices some of the array gain to reduce interference. Fig.~\ref{Fig-AverageMaxMinQoS} and Fig.~\ref{Fig-AverageMaxMinQoS40} also show that, for example, a system with $200$ BS antennas and using non-coherent joint transmission can serve up to $20$ users and $40$ users for the QoS requirement of $2.28$ (bit/symbol) and $1.87$ (bit/symbol) respectively.

\begin{figure}[t]
 \begin{minipage}{0.48\textwidth}
       \centering
       \includegraphics[trim=2.0cm 0.9cm 2cm 1.5cm, clip=true, width=3 in]{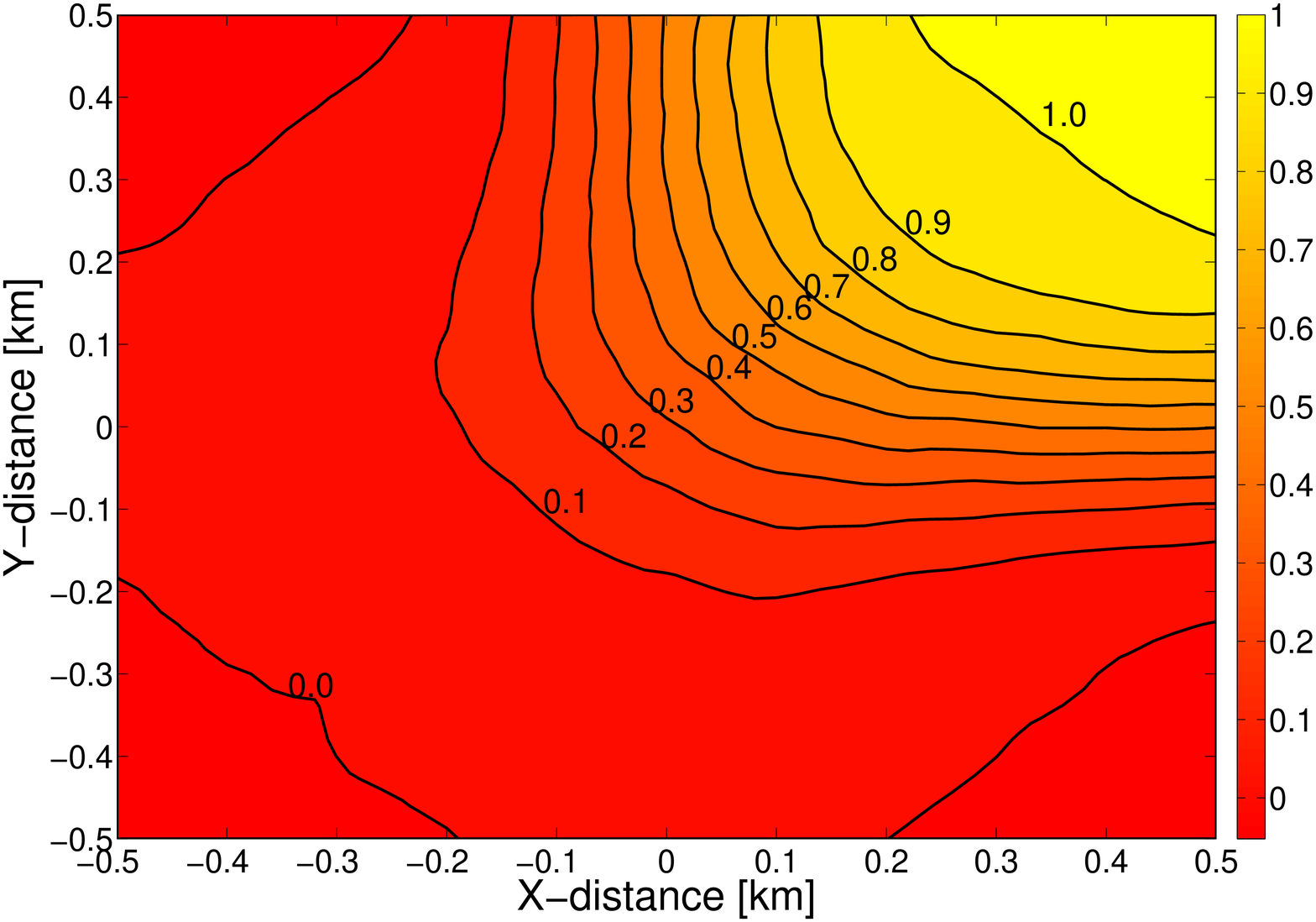}
       \caption{The probability that a user is served by BS $1$ for the max-min QoS algorithm with MRT precoding and $M =200, K = 40$.}
       \label{Fig-MRTContour}
      
 \end{minipage}
 \hfill
 \begin{minipage}{0.48\textwidth}
       \centering
       \includegraphics[trim=2.0cm 0.9cm 2cm 1.5cm, clip=true, width=3 in]{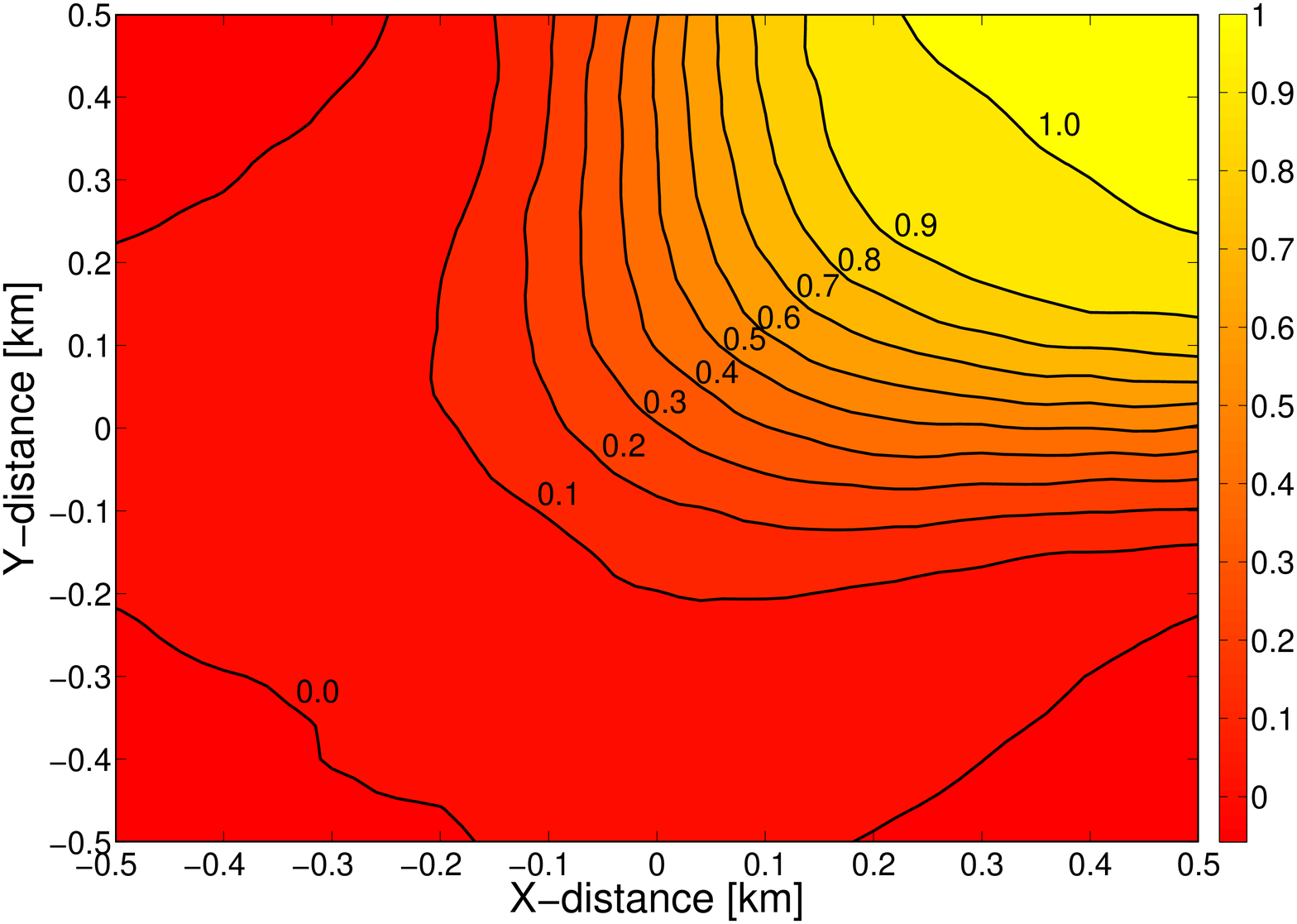}
       \caption{The probability that a user is served by BS $1$ for the max-min QoS algorithm with ZF precoding and $M =200 , K = 40$.}
       \label{Fig-ZFContour}
     
   \end{minipage}
 \end{figure}
 
The probability that a user is served by more than one BS is shown in Fig.~\ref{Fig-BSAssociationProbability}. For pair comparison, we consider $20$ users for both fixed QoS and max-min QoS. Although the system model lets BSs cooperate with each other to serve the users, experimental results verify that single-BS association is enough in $90 \%$ or more of the cases. This result for multi-cell Massive MIMO systems is similar to those obtained by multi-tier heterogeneous network with multiple-antenna \cite{Li2015} or single-antenna macrocell BSs \cite{Ye2013a}. In the remaining cases corresponding to Case $2$ in Theorem \ref{Theorem-BS-Association}, multiple BSs are required to deal with severe shadow fading realizations or high user loads.

Fig.~\ref{Fig-MRTContour} and Fig.~\ref{Fig-ZFContour} show the probabilities of users being associated with BS $1$ located at the coordinate $(0.5, 0.5)$ as a function of user locations. Intuitively, users near the BS in the sense of physical distance tend to associate with high probability. For example, most user locations that have their coordinates $(X > 0.1, Y > 0.1)$ are served by this BS with a probability larger than $0.5$. In contrast, the users that lie around the origin are only served by BS $1$ with probability less than $0.4$ and they are likely to associate with multiple BSs. We also observe that BS $1$ associates with some very distant users (i.e., they are not located in Quadrant $1$ as shown in Fig.~\ref{fig:MassiveMIMOSystem-Layout}). These situations occur due to severe shadow fading realizations or due to high user loads which make the closest BS not be the best selection.

\section{Conclusion} \label{Section:Conclusion}
This paper proposed a new method to jointly optimize the power allocation and user association in multi-cell Massive MIMO systems. The DL non-coherent joint transmission was designed to minimize the total transmit power consumption while satisfying QoS constraints. For Rayleigh fading channels, we proved that the total transmit power minimization problem with MRT or ZF precoding is a linear program, so it is always solvable to global optimality in polynomial time. Additionally, we provided the optimal BS-user association rule to serve the users. In order to ensure that all users are fairly treated, we also solved the weighted max-min fairness optimization problem which maximizes the worst QoS value among the users with user specific weighting. Experimental results manifested the effectiveness of our proposed methods, and that the max-SNR association works well in many Massive MIMO scenarios but is not optimal.
    
\appendix

\subsection{Proof of Proposition \ref{Proposition: Rate} and Theorem \ref{Theorem-Lower-Bound-Rate} } \label{Appendix Proposition: Rate} 
When user $k$ detects the signal from BS $1$, it will not know any of the transmitted signals. Therefore, the received signal in \eqref{eq: Downlink-Signal} is expressed as
\begin{equation} \label{eq: Signal-s_1k}
\begin{split}
& y_{1,k}  = y_k = \sqrt{ \rho_{1,k} } \mathbb{E} \{ \mathbf{h}_{1,k}^H \mathbf{w}_{1,k} \} s_{1,k} \\
& \; \; \; +  \sqrt{ \rho_{1,k} } \left(\mathbf{h}_{1,k}^H \mathbf{w}_{1,k} - \mathbb{E} \{ \mathbf{h}_{1,k}^H \mathbf{w}_{1,k} \} \right) s_{1,k} \\
& \; \; \; + \sum_{ i =2}^{L}  \sqrt{ \rho_{i,k} } \mathbf{h}_{i,k}^H \mathbf{w}_{i,k} s_{i,k} +
\sum_{i =1}^{L} \sum_{ \substack{t=1 \\ t \neq k} }^{K} \sqrt{ \rho_{i,t} } \mathbf{h}_{i,k}^H \mathbf{w}_{i,t} s_{i,t}+ n_k.
\end{split} 
\end{equation}
At the point when user $k$ detects the signal coming from BS $l$, for $l = 2, \ldots, L$, it possesses the set of all detected signals coming from the first $(l-1)$ BSs. The received signal in \eqref{eq: Downlink-Signal} is processed by subtracting the known signals over the known average channels as
\begin{equation} \label{eq: Signal-s_lk}
\begin{split}
& y_{l,k}  = y_k - \sum_{i=1}^{l-1} \sqrt{ \rho_{i,k} } \mathbb{E} \{ \mathbf{h}_{i,k}^H \mathbf{w}_{i,k} \} s_{i,k} = \sqrt{ \rho_{l,k} } \mathbb{E} \{ \mathbf{h}_{l,k}^H \mathbf{w}_{l,k} \} s_{l,k} \\
& \; \; \;+ \sqrt{ \rho_{l,k} } \left( \mathbf{h}_{l,k}^H \mathbf{w}_{l,k} - \mathbb{E} \{ \mathbf{h}_{l,k}^H \mathbf{w}_{l,k} \} \right) s_{l,k} \\
 & \; \; \;+ \sum_{ i = l +1 }^{L} \sqrt{ \rho_{i,k} } \mathbf{h}_{i,k}^H \mathbf{w}_{i,k} s_{i,k}  \\
& \; \; \; + \sum_{i=1}^{l-1} \sqrt{ \rho_{i,k} } ( \mathbf{h}_{i,k}^H \mathbf{w}_{i,k}  -  \mathbb{E} \{ \mathbf{h}_{i,k}^H \mathbf{w}_{i,k} \} )  s_{i,k}  \\
 &\; \;  \; +
  \sum_{ i = 1 }^{L} \sum_{ \substack{ t =1 \\ t \neq k} }^{K} \sqrt{ \rho_{i,t} }  \mathbf{h}_{i,k}^H \mathbf{w}_{i,t} s_{i,t} + n_k.
\end{split}
\end{equation}

In \eqref{eq: Signal-s_1k} and \eqref{eq: Signal-s_lk}, we respectively add-and-subtract the terms $\sqrt{ \rho_{1,k} } \mathbb{E} \{ \mathbf{h}_{1,k}^H \mathbf{w}_{1,k} \} s_{1,k}$ and  $\sqrt{ \rho_{l,k} } \mathbb{E} \{ \mathbf{h}_{l,k}^H \mathbf{w}_{l,k} \} s_{l,k}$. As a result, the first term after their second equality contains the desired signal from BS $l$ that is now transmitted over a deterministic channel while other terms are treated as uncorrelated noise. A lower bound on the ergodic capacity $C_{l,k}$ of the transmission from BS $l$ is obtained by considering Gaussian noise as the worst case distribution of the uncorrelated noise \cite{Hassibi2003a},
\begin{equation}
C_{l,k} \geq \gamma^{\textrm{DL}} \left(1-\frac{\tau_p}{\tau_c} \right) \log_2 \left(1 + \frac{ \mathbb{E} \{ |\mathtt{DS}_{l,k}|^2 \}}{\mathbb{E} \{ |\mathtt{UN}_{l,k}|^2 \}} \right),
\end{equation}
where the desired signal power $\mathbb{E} \{ |\mathtt{DS}_{l,k}|^2 \}$ is computed as
\begin{equation} \label{eq: NumeratorSINR}
\mathbb{E} \{ |\mathtt{DS}_{l,k}|^2 \} = \rho_{l,k} | \mathbb{E} \{ \mathbf{h}_{l,k}^H \mathbf{w}_{l,k} \} |^2
\end{equation}
and the uncorrelated noise power $\mathbb{E} \{ |\mathtt{UN}_{l,k}|^2 \}$ is computed as
\begin{equation} \label{eq: DenominatorSINR}
\begin{split}
& \mathbb{E} \{ |\mathtt{UN}_{l,k}|^2 \} =  \rho_{l,k}  \mathbb{E} \{ | \mathbf{h}_{l,k}^H \mathbf{w}_{l,k} - \mathbb{E} \{ \mathbf{h}_{l,k}^H \mathbf{w}_{l,k} \} |^2 \}\\
& \; \; \; + \sum_{ i = l +1 }^{L} \rho_{i,k} \mathbb{E} \{  |\mathbf{h}_{i,k}^H \mathbf{w}_{i,k}|^2 \} \\
& \; \; \; + 
 \sum_{i=1}^{l-1}  \rho_{i,k} \mathbb{E} \{ | \mathbf{h}_{i,k}^H \mathbf{w}_{i,k}  -  \mathbb{E} \{ \mathbf{h}_{i,k}^H \mathbf{w}_{i,k} \} |^2 \} \\
 &\; \; \; + 
\sum_{ i = 1 }^{L} \sum_{ \substack{ t =1 \\ t \neq k} }^{K} \rho_{i,t} \mathbb{E} \{ |\mathbf{h}_{i,k}^H \mathbf{w}_{i,t} |^2  \} + \sigma_{\textrm{DL}}^2\\
&=\sum_{ i=1 }^{L} \sum_{ t =1 }^{K} \rho_{i,t} \mathbb{E} \{ | \mathbf{h}_{i,k}^H \mathbf{w}_{i,t} |^2 \}  -  \sum_{ i =1}^{l}  \rho_{i,k} | \mathbb{E} \{ \mathbf{h}_{i,k}^H \mathbf{w}_{i,k} \} |^2  + \sigma_{ \mathrm{DL} }^2. 
\end{split}
\end{equation}
By letting $\textrm{SINR}_{l,k} = \frac{ \mathbb{E} \{ |\mathtt{DS}_{l,k}|^2 \}}{\mathbb{E} \{ |\mathtt{UN}_{l,k}|^2 \}} $, and then subtracting \eqref{eq: NumeratorSINR} and \eqref{eq: DenominatorSINR} into the SINR value we obtain the DL ergodic rate between each BS and user $k$ as stated in Proposition \ref{Proposition: Rate}.

We have proved Proposition \ref{Proposition: Rate} and will detect the $L$ signals in a successive manner to prove Theorem \ref{Theorem-Lower-Bound-Rate}. Consequently, a lower bound on the sum SE of user $k$ is obtained by
\begin{equation}
\begin{split}
& R_k = \sum_{l =1 }^{L} R_{l ,k} = \gamma^{\textrm{DL}} \left(1-\frac{\tau_p}{\tau_c} \right) \log_2 \left( \prod_{l=1}^{L}  \underbrace{( 1 + \textrm{SINR}_{l,k} )}_{=  \mathcal{A}_{l,k} } \right),
\end{split}
\end{equation}
where $\mathcal{A}_{l,k}$ is given as
  \begin{equation} \label{eq: A_lk}
 \frac{\sum\limits_{i = 1}^{L} \sum\limits_{t =1}^{K} \rho_{i,t} \mathbb{E} \{ | \mathbf{h}_{i,k}^H \mathbf{w}_{i,t} |^2 \} - \sum\limits_{ i=1}^{l-1}  \rho_{i,k} | \mathbb{E} \{ \mathbf{h}_{i,k}^H \mathbf{w}_{i,k} \} |^2 + \sigma_{ \mathrm{DL} }^2 }{  \sum\limits_{i = 1}^{L} \sum\limits_{t =1}^{K} \rho_{i,t} \mathbb{E} \{ | \mathbf{h}_{i,k}^H \mathbf{w}_{i,t} |^2 \} - \sum\limits_{ i=1}^{l}  \rho_{i,k} | \mathbb{E} \{ \mathbf{h}_{i,k}^H \mathbf{w}_{i,k} \} |^2 + \sigma_{ \mathrm{DL} }^2 }.
  \end{equation}
It is observed that the denominator of $\mathcal{A}_{l,k}$ coincides with the numerator of $\mathcal{A}_{l+1,k},$ for $l =1, \ldots, L-1$. Thus, after some manipulation which cancels out these coincided terms, we obtain
\begin{equation}
\prod_{l=1}^{L} \mathcal{A}_{l,k} = \frac{\mathtt{A}_{\mathtt{Num}}}{\mathtt{A}_{\mathtt{Den}}},
\end{equation}
where the values $\mathtt{A}_{\mathtt{Num}}$ and $\mathtt{A}_{\mathtt{Den}}$ are defined as
\begin{equation*}
\mathtt{A}_{\mathtt{Num}} = \sum_{i = 1}^{L} \sum_{t =1}^{K} \rho_{i,t} \mathbb{E} \{ | \mathbf{h}_{i,k}^H \mathbf{w}_{i,t} |^2 \} + \sigma_{ \mathrm{DL} }^2, 
\end{equation*}
\begin{equation*}
\mathtt{A}_{\mathtt{Den}} =  \sum_{i = 1}^{L} \sum_{t =1}^{K} \rho_{i,t} \mathbb{E} \{ | \mathbf{h}_{i,k}^H \mathbf{w}_{i,t} |^2 \} - \sum_{ i=1}^{L}  \rho_{i,k} | \mathbb{E} \{ \mathbf{h}_{i,k}^H \mathbf{w}_{i,k} \} |^2 + \sigma_{ \mathrm{DL} }^2. 
\end{equation*}
By simplifying the ratio $\mathtt{A}_{\mathtt{Num}}/\mathtt{A}_{\mathtt{Deno}} $, $R_k$ is given as \eqref{eq: Sum-Rate-k} in the theorem.


\subsection{Proof of Corollary~\ref{Corollary-MRT-Rate}}
 \label{Appendix Corollary-MRT-Rate}
Because the channels are Rayleigh fading, the expected squared norm of the channel between BS $i$ and user $t$ is
\begin{equation} \label{eq: Second-Moment-Channel}
\mathbb{E} \{ \| \hat{ \mathbf{h} }_{i,t} \|^2  \}  = M \theta_{i,t}.
\end{equation}
Combining \eqref{eq: Second-Moment-Channel} and \eqref{eq: Linear-Precoding-Vector}, the MRT precoding vector $ \mathbf{w}_{i,t}$ is
\begin{equation}
  \mathbf{w}_{i,t} = \frac{1}{\sqrt{M \theta_{i,t}}} \hat{ \mathbf{h} }_{i,t}.
\end{equation}
Since the estimation error is independent of the corresponding estimate, the numerator in \eqref{eq: SINR_k} is
\begin{equation} \label{eq:MRT-Numerator}
\begin{split}
&\sum_{i=1}^{L} \rho_{i,k} | \mathbb{E} \{  \mathbf{h}_{i,k}^H  \mathbf{w}_{i,k}  \} |^2 = 
M  \sum_{i =1}^{L} \rho_{i,k} \theta_{i,k}.
\end{split}
\end{equation}
In addition, we reformulate the denominator in \eqref{eq: SINR_k} as
\begin{equation} \label{eq:Denominator_new}
\sum_{i=1}^L \sum_{t=1}^K \rho_{i,t} \mathbb{E} \{ | \mathbf{h}_{i,k}^H \mathbf{w}_{i,t}|^2 \} - \sum_{i=1}^L \rho_{i,k} |\mathbb{E} \{ \mathbf{h}_{i,k}^H \mathbf{w}_{i,k} \} |^2 + \sigma_{\textrm{DL}}^2.
\end{equation}
The first summation of the denominator in \eqref{eq:Denominator_new} is decomposed into two parts based on the pilot reuse set $\mathcal{P}_k$ as follows
\begin{equation} \label{eq: MRT-Denominator1}
\begin{split}
&\sum_{i=1}^{L} \sum_{t=1}^K \rho_{i,t} \mathbb{E} \{ | \mathbf{h}_{i,k}^H \mathbf{w}_{i,t} |^2 \} \\
&= 
\sum_{i=1}^L \sum_{t \in \mathcal{P}_k} \rho_{i,t} \mathbb{E} \{ |\mathbf{h}_{i,k}^H \mathbf{w}_{i,t}|^2 \} + \sum_{i=1}^L \sum_{t \notin \mathcal{P}_k} \rho_{i,t} \mathbb{E} \{ |\mathbf{h}_{i,k}^H \mathbf{w}_{i,t}|^2 \} \\
&=  
\sum_{i=1}^L \sum_{t \in \mathcal{P}_k} \rho_{i,t} \mathbb{E} \{ |\hat{\mathbf{h}}_{i,k}^H \mathbf{w}_{i,t}|^2\} + \sum_{i=1}^L \sum_{t \in \mathcal{P}_k} \rho_{i,t} \mathbb{E} \{|\mathbf{e}_{i,k}^H \mathbf{w}_{i,t}|^2\} \\
& \; \; \; + \sum_{i=1}^L \sum_{t \notin \mathcal{P}_k } \rho_{i,t} \beta_{i,k} \\
&\stackrel{(a)}{=} \sum_{i=1}^L \sum_{t \in \mathcal{P}_k} \rho_{i,t} \frac{p_k \beta_{i,k}^2 }{p_t \beta_{i,t}^2} \mathbb{E} \{ |\hat{\mathbf{h}}_{i,t}^H \mathbf{w}_{i,t}|^2\} \\
& \; \; \; + \sum_{i=1}^L \sum_{t \in \mathcal{P}_k} \rho_{i,t} \left( \beta_{i,k} - \theta_{i,k} \right) 
 + \sum_{i=1}^L \sum_{t \notin \mathcal{P}_k } \rho_{i,t} \beta_{i,k} \\
& \stackrel{(b)}{=} M \sum_{i=1}^L \sum_{t \in \mathcal{P}_k} \rho_{i,t} \theta_{i,k}  + \sum_{i=1}^L \sum_{t =1}^K \rho_{i,t} \beta_{i,k} .
\end{split}
\end{equation}
In \eqref{eq: MRT-Denominator1},  the relationship between the channel estimates of two users utilizing the same pilot sequences as stated in \eqref{eq:Channel_Relationship} is used to compute $(a)$. For $(b)$, we use Lemma 2.9 in \cite{Tulino2004} to compute the fourth-order moment $\mathbb{E} \{ || \hat{ \mathbf{h} }_{i,t} ||^4 \}$.  The denominator in \eqref{eq: SINR_k}  is obtained by plugging \eqref{eq:MRT-Numerator} and \eqref{eq: MRT-Denominator1} into \eqref{eq:Denominator_new}.  Combining this denominator and the numerator in \eqref{eq:MRT-Numerator}, the SINR value is shown in the corollary.
\vspace*{-0.3cm}
\subsection{Proof of Corollary~\ref{Corollary-ZF-Rate}} \label{Appendix Corollary-ZF-Rate}
\vspace*{-0.2cm}
By utilizing Lemma 2.10 in \cite{Tulino2004} for a $K \times K$ central complex Wishart matrix with $M$ degrees of freedom which satisfies $M \geq K+1$, we obtain
\begin{equation}
\mathbb{E} \{ \| \hat{\mathbf{H}}_i \mathbf{r}_{i,t} \|^2 \} = \mathbb{E} \{ [ \hat{\mathbf{H}}_i^H \hat{\mathbf{H}}_i ]^{-1}_{t,t} \} =
\frac{1}{ (M-K)\theta_{i,t}}.
\end{equation}
Hence, the ZF precoding vector $\mathbf{w}_{i,t}$ becomes
\begin{equation} \label{eq: ZF-Precoding-VectorProof}
\mathbf{w}_{i,t} =  \sqrt{ (M-K) \theta_{i,t}} \hat{\mathbf{H}}_i \mathbf{r}_{i,t}. 
\end{equation}
Combining the result in \eqref{eq: ZF-Precoding-VectorProof}, the ZF properties in \eqref{eq: ZF-Property}, and the independence between channel estimates and estimation errors, the numerator of \eqref{eq: SINR_k} becomes
\begin{equation} \label{eq: ZF-Numerator}
\sum_{i =1 }^{L} \rho_{i,k} \left| \mathbb{E} \left\{ \mathbf{h}_{i,k}^H \mathbf{w}_{i,k} \right\} \right|^2 = (M-K) \sum_{ i =1}^{L} \rho_{i,k} \theta_{i,k} .
\end{equation}
Similarly, the first part of the denominator in \eqref{eq:Denominator_new} is
\begin{equation} \label{eq: ZF-Denominator1}
\begin{split}
& \sum_{i=1}^{L} \sum_{t=1}^K \rho_{i,t} \mathbb{E} \{ | \mathbf{h}_{i,k}^H \mathbf{w}_{i,t} |^2 \} \\
&= 
\sum_{i=1}^L \sum_{t \in \mathcal{P}_k} \rho_{i,t} \mathbb{E} \{ |\hat{\mathbf{h}}_{i,k}^H \mathbf{w}_{i,t}|^2 \} + \sum_{i=1}^L \sum_{t =1}^K \rho_{i,t} \mathbb{E} \{ |\mathbf{e}_{i,k}^H \mathbf{w}_{i,t}|^2 \} \\
&= 
\sum_{i=1}^L \sum_{t \in \mathcal{P}_k} \rho_{i,t} \frac{p_k \beta_{i,k}^2 }{p_t \beta_{i,t}^2  }\mathbb{E} \{ |\hat{\mathbf{h}}_{i,t}^H \mathbf{w}_{i,t}|^2 \} + \sum_{i=1}^L \sum_{t =1}^K \rho_{i,t} \left( \beta_{i,k} - \theta_{i,k} \right) \\
&= (M-K) \sum_{i=1}^L \sum_{t \in \mathcal{P}_k} \rho_{i,t} \theta_{i,k} + \sum_{i=1}^L \sum_{t =1}^K \rho_{i,t} \left( \beta_{i,k} - \theta_{i,k} \right) .
\end{split}
\end{equation}
Combining \eqref{eq: ZF-Numerator} and \eqref{eq: ZF-Denominator1}, the denominator of \eqref{eq: SINR_k} is 
\begin{equation} \label{eq: ZF-Denominator}
(M-K) \sum_{i=1}^L \sum_{t \in \mathcal{P}_k \setminus \{k\} } \rho_{i,t} \theta_{i,k} + \sum_{i=1}^L \sum_{t =1}^K \rho_{i,t} \left( \beta_{i,k} - \theta_{i,k} \right) + \sigma_{\textrm{DL}}^2.
\end{equation}
Plugging \eqref{eq: ZF-Numerator} and \eqref{eq: ZF-Denominator} to \eqref{eq: SINR_k}, we get the SINR value as shown in the corollary. 

\subsection{Proof of Theorem~\ref{Theorem-BS-Association}} \label{Appendix BS-Association}
To prove this result, we first make a change of variable to  $\mathbf{u}_t =[\sqrt{\rho_{1,t}},\ldots, \sqrt{\rho_{L,t}} ]^T \in \mathbb{C}^{L}$ and define the diagonal matrix $\mathbf{A}_t = \mathrm{diag}(a_{1,t}, \ldots, a_{L,t} ) \in \mathbb{C}^{L \times L}$, where $a_{i,t}$, for $i = 1, \ldots, L$ are elements of $\mathbf{a}_t$. The Lagrangian in \eqref{eq: Langrangian} is then converted to a quadratic function
\begin{equation} \label{eq: Lagrangian-Proof}
\begin{split}
\mathcal{L} \left( \mathbf{u}_t, \lambda_k, \mu_i \right) &= \sum_{k=1}^{K} \lambda_k - \sum_{i=1}^{L} \mu_i P_{ \mathrm{max},i} +  \sum_{t=1}^{K} \mathbf{u}_t^T  \mathbf{A}_t \mathbf{u}_t.
\end{split}
\end{equation}
The Lagrange dual function of \eqref{eq: Lagrangian-Proof} is further formulated as
\begin{equation} \label{eq: Langrangian-Duality-Proof}
\begin{split}
&\mathcal{G}\left(\lambda_k, \mu_i \right) = \underset{  \{ \pmb{\rho}_t \} }{ \inf} \; \mathcal{L} \left( \mathbf{u}_t, \lambda_k, \mu_i \right) \\
&=
\sum_{k=1}^{K} \lambda_k \sigma_{ \mathrm{DL} }^2 - \sum_{i=1}^{L} \mu_i P_{\mathrm{max},i} + \underset{ \{ \mathbf{u}_t \} }{ \inf } \;  \sum_{t=1}^{K} \mathbf{u}_t^T  \mathbf{A}_t \mathbf{u}_t.
\end{split}
\end{equation}
Therefore, $\mathcal{G} \left( \lambda_k, \mu_i \right)$ is bounded from below if and only if $\mathbf{A}_t \succeq  0$. Taking the first-order derivative of the Lagrangian in \eqref{eq: Lagrangian-Proof} with respect to $\mathbf{u}_t$, we obtain
\begin{equation} \label{eq: First-Derivative}
2 \check{ \mathbf{A} }_t \check{ \mathbf{u} }_t = \mathbf{0},
\end{equation}
where $ \check{\mathbf{A}}_t $ and $\check{\mathbf{u}}_t $ are the optimal solutions of $\mathbf{A}_t$ and $\mathbf{u}_t$, respectively.
Hence, \eqref{eq: First-Derivative} gives the following $L$ necessary and sufficient conditions
\begin{equation} \label{eq: First-Derivative1}
\begin{split}
&\sqrt{ \rho_{i,t} } \left( \Delta_i+ \sum_{k=1}^K \lambda_k \theta_{i,k} \mathbbm{1}_k (t)  + \sum_{k=1}^{K} \lambda_k c_{i,k} - \lambda_t b_{i,t} + \sum_{i=1}^L \mu_i \right)\\
&
= 0. 
\end{split}
\end{equation}
where $c_{i,k}$ and $b_{i,t}$ are the $i$th entry of the vectors $\mathbf{c}_k$ and $\mathbf{b}_t$, respectively which are defined in Theorem \ref{Theorem: Linear-Solution}. If BS $i$ associates with user $t$ (i.e., $\rho_{i,t} \neq 0 $), then from \eqref{eq: First-Derivative} we achieve
\begin{equation} \label{eq: Optimal-Lagrangian}
 \Delta_i+ \sum_{k=1}^K \lambda_k \theta_{i,k} \mathbbm{1}_k (t) + \sum_{k=1}^{K} \lambda_k c_{i,k} - \lambda_t b_{i,t} + \sum_{i=1}^L \mu_i = 0,
\end{equation}
from which the optimal Lagrange multiplier $\check{\lambda}_{t}$ is
\begin{equation} \label{eq: Max-Lambda}
\check{\lambda}_{t} = \left( \Delta_i + \sum_{k=1}^{K} \lambda_k \theta_{i,k} \mathbbm{1}_k (t) + \sum\limits_{k=1}^{K} \check{\lambda}_{k} c_{i,k} + \sum\limits_{i=1}^{L} \check{\mu}_{i} \right)\frac{ 1 }{b_{i,t}}. 
\end{equation}
According to the duality regularization, $\check{\mathbf{A}} \succeq 0$, we stress that
\begin{equation}
\lambda_t \leq \left( \Delta_i + \sum_{k=1}^{K} \check{\lambda}_k \theta_{i,k} \mathbbm{1}_k (t) + \sum\limits_{k=1}^{K} \check{\lambda}_{k} c_{i,k} + \sum\limits_{i=1}^{L} \check{\mu}_{i} \right)\frac{ 1 }{b_{i,t}}.
\end{equation} 
The above equation implies that the system only selects BSs satisfying \eqref{eq: Max-Lambda}, otherwise transmit powers are set to zero due to \eqref{eq: First-Derivative1} and hence there is no communication between these BSs and user $t$. Moreover, the QoS constraints ensure that user $t$ must be served by at least one BS. The BS association of user $t$ is hence defined as shown in the theorem. 
\vspace*{-0.5cm}
\subsection{Proof of Corollary ~\ref{Corollary-Bisection}}
 \label{Appendix Bisection}
 \vspace*{-0.2cm}
 Because pilot reuse reduces the SE of the users, we only need to estimate $\xi_0^{\mathrm{upper}}$ for the optimistic special case all the users use  mutually orthogonal pilot sequences, and then this upper bound also applies for the scenario where the system suffers from pilot contamination effects. From Theorem \ref{Theorem-Bisection}, $\xi_0^{\mathrm{upper}}$ can be computed as
\begin{equation} \label{eq: UpperSINR-Definition}
\xi_0^{ \mathrm{upper} } \triangleq \underset{k}{ \min } \frac{ R_k }{ w_k } = \underset{k}{ \min } \; \gamma^{\mathrm{DL}} \left( 1 - \frac{\tau_p}{\tau_c} \right) \log_2 \left( 1 + \mathrm{SINR}_k \right) .
\end{equation}
To solve \eqref{eq: UpperSINR-Definition}, we first compute the maximal SINR value. In the case of MRT precoding, from  \eqref{eq: SINR-MRTOrthogonal} we observe
\begin{equation} \label{eq: SINR-MRT-Proof}
\begin{split}
\mathrm{SINR}_k^{ \mathrm{MRT}} &= \frac{ M \sum\limits_{i =1}^{L} \frac{\rho_{i,k} p_k \tau_p \beta_{i,k}^2 }{p_k \tau_p \beta_{i,k} + \sigma_{\textrm{UL}}^2 } }{\sum\limits_{i = 1 }^{L} \sum\limits_{ t=1 }^{K} \rho_{i,t} \beta_{i,k} + \sigma_{ \mathrm{DL}}^2 } \\
& \stackrel{(a)}{ \leq }\frac{ M \sum\limits_{i =1}^{L} \rho_{i,k} \beta_{i,k } }{ \sum\limits_{i = 1 }^{L} \sum\limits_{ t=1 }^{K} \rho_{i,t} \beta_{i,k} + \sigma_{ \mathrm{DL}}^2 } \stackrel{(b)}{ \leq } M.
 \end{split}
\end{equation}
In \eqref{eq: SINR-MRT-Proof}, $(a)$ is because $ \frac{p_k \tau_p \beta_{i,k} }{ p_k \tau_p \beta_{i,k} + \sigma_{\mathrm{UL}}^2 } \leq 1 $ and $(b)$ is obtained since $\sum_{i =1}^{L} \rho_{i,k} \beta_{i,k } \leq (  \sum_{i = 1 }^{L} \sum_{ t=1 }^{K} \rho_{i,t} \beta_{i,k} + \sigma_{ \mathrm{DL}}^2  )$. Combining \eqref{eq: UpperSINR-Definition} and \eqref{eq: SINR-MRT-Proof}, $\xi_0^{\mathrm{upper}}$ is selected as in the corollary.

In the case of ZF precoding, we first obtain
\begin{equation} \label{eq: Numerator-ZF1}
\begin{split}
\sum\limits_{i =1}^{L} \left( \frac{ \rho_{i,k} \beta_{i,k}}{ p_k \tau_p \beta_{i,k} + \sigma_{\mathrm{UL}}^2 } \right)\beta_{i,k} \leq  \sum\limits_{i =1}^{L} \frac{ \rho_{i,k} \beta_{i,k}}{ p_k \tau_p \beta_{i,k} + \sigma_{\mathrm{UL}}^2 }  \sum\limits_{i =1}^{L} \beta_{i,k}
\end{split}
\end{equation}
by utilizing the Cauchy-Schwarz's inequality and the facts that $\sum_{i =1}^{L}  ( \frac{ \rho_{i,k} \beta_{i,k}}{ p_k \tau_p \beta_{i,k} + \sigma_{\mathrm{UL}}^2 } )^2 \leq ( \sum_{i =1}^{L} \frac{ \rho_{i,k} \beta_{i,k}}{ p_k \tau_p \beta_{i,k} + \sigma_{\mathrm{UL}}^2 } )^2 $ along with $\sum_{i =1}^{L} \beta_{i,k}^2  \leq  ( \sum_{i =1}^{L} \beta_{i,k} )^2$. Consequently, the SINR value can be upper bounded as
\begin{equation} \label{eq: SINR-ZF-Proof}
\begin{split}
\mathrm{SINR}_k^{\mathrm{ZF}} &= \frac{ (M -K) \sum\limits_{i =1}^{L} \frac{\rho_{i,k} p_k \tau_p \beta_{i,k}^2 }{p_k \tau_p \beta_{i,k} + \sigma_{\textrm{UL}}^2 } }{\sum\limits_{i = 1 }^{L} \sum\limits_{ t=1 }^{K}  \frac{ \rho_{i,t} \beta_{i,k} \sigma_{\textrm{UL}}^2}{p_k \tau_p \beta_{i,k} + \sigma_{\textrm{UL}}^2} + \sigma_{ \mathrm{DL}}^2 } \\
& \leq \frac{ (M-K) \frac{p_k \tau_p}{\sigma_{\mathrm{UL}}^2}   \sum\limits_{i =1}^{L} \beta_{i,k} \sum\limits_{i =1}^{L} \frac{ \rho_{i,k} \beta_{i,k} \sigma_{\mathrm{UL}}^2 }{ p_k \tau_p \beta_{i,k} + \sigma_{\mathrm{UL}}^2 }  }{ \sum\limits_{ i =1 }^{L} \sum\limits_{ t=1 }^{K} \rho_{i,t} \frac{\beta_{i,k} \sigma_{ \mathrm{UL}}^2 }{p_k \tau_p \beta_{i,k} + \sigma_{ \mathrm{UL}}^2 } + \sigma_{ \mathrm{DL}}^2 } \\
&\leq  \frac{(M-K) p_k \tau_p }{ \sigma_{\mathrm{UL}}^2 }  \sum_{i =1}^{L} \beta_{i,k}. 
\end{split}
\end{equation}
In summary, combining \eqref{eq: UpperSINR-Definition} and \eqref{eq: SINR-ZF-Proof}, $\xi_0^{\mathrm{upper}}$ can be selected as stated in the corollary.

\subsection{Joint Power Allocation and User Association for Massive MIMO Systems with Coherent Joint Transmission}

 \label{Appendix: Coherent Joint Transmission}
 With coherent joint transmission all BSs in the network will precode and send the same signal to a user. It means that the received signal at user $k$ is
 \begin{equation}
y_k = \sum_{i=1}^{L} \sqrt{\rho_{i,k}} \mathbf{h}_{i,k}^H \mathbf{w}_{i,k} s_{k} + \sum_{i=1}^{L} \sum_{ \substack{ t=1 \\ t \neq k}}^{K} \sqrt{\rho_{i,t}} \mathbf{h}_{i,k}^H \mathbf{w}_{i,t} s_t + n_k. 
 \end{equation}
 Applying the added-and-subtract technique that is shown in \eqref{eq: Signal-s_1k} and \eqref{eq: Signal-s_lk} and then considering Gaussian noise as the worst case distribution of the uncorrelated noise \cite{Hassibi2003a}, a lower bound on the ergodic SE of user $k$ is obtained as
 \begin{equation}
 R_k = \gamma^{\mathrm{DL}} \left( 1 - \frac{\tau_p}{\tau_c}  \right) \log_2 (1 + \mathrm{SINR}_k ) \quad \textrm{[bit/symbol]},
 \end{equation}
where the SINR value, $\textrm{SINR}_{k}$, is presented as
   \begin{equation} \label{eq: Rate_Coherent}
 \frac{  \left| \sum\limits_{i=1}^L \sqrt{\rho_{i,k}} \mathbb{E} \{ \mathbf{h}_{i,k}^H \mathbf{w}_{i,k} \} \right|^2  }{ \sum\limits_{t=1}^K \mathbb{E} \left\{ \left| \sum\limits_{i=1}^L \sqrt{\rho_{i,t}} \mathbf{h}_{i,k}^H \mathbf{w}_{i,t} \right|^2 \right\}  - \left| \sum\limits_{i=1}^L \sqrt{\rho_{i,k}} \mathbb{E} \{ \mathbf{h}_{i,k}^H \mathbf{w}_{i,k} \} \right|^2   + \sigma_{\textrm{DL}}^2 }.
   \end{equation}
Utilizing the same techniques as in Appendix \ref{Appendix Corollary-MRT-Rate} and \ref{Appendix Corollary-ZF-Rate}, the total transmit power minimization problem is expressed for Rayleigh fading channels together with MRT or ZF precoding
\begin{equation} \label{eq: Optimization_Coherent} 
 \begin{aligned}
 & \underset{ \{ \rho_{i,t} \geq 0 \}}{\mathrm{minimize}} & &
 \sum_{i=1}^{L} \Delta_i \sum_{t=1}^{K} \rho_{i,t}  \\
 & \mbox{subject to} && \frac{\left( \sum\limits_{i=1}^L \sqrt{ \rho_{i,k} g_{i,k} } \right)^2 }{\sum\limits_{i=1}^L \sum\limits_{t \in \mathcal{P}_{k} \setminus \{k\} } \rho_{i,t} g_{i,k}+ \sum\limits_{ i=1 }^L   \sum\limits_{t=1}^K \rho_{i,t} z_{i,k} + \sigma_{\textrm{DL} }^2 }  \\
 &&& \geq \hat{ \xi }_k, \; \forall k \\
 &&& \sum_{t=1}^{K} \rho_{i,t} \leq P_{\mathrm{max},i}, \; \forall i.  \\
 \end{aligned}
 \end{equation}
 Here, the parameters $g_{i,k}$ and $z_{i,k}$ are specified by the precoding scheme. MRT precoding gives $g_{i,k} =  M \theta_{i,k}  $ and $z_{i,k} = \beta_{i,k}$ while ZF precoding obtains $g_{i,k} =  (M-K) \theta_{i,k}  $ and $z_{i,k} = \beta_{i,k} - \theta_{i,k}$. Let $\mathbf{U} = [\mathbf{u}_1, \ldots , \mathbf{u}_K ] \in \mathbb{C}^{M \times K}$ have columns $\mathbf{u}_{t} = [\sqrt{\rho_{1,t}}, \ldots, \sqrt{\rho_{L,t}}]^T, \mbox{ for } t= 1, \ldots, K$.  Therefore, we can denote $\mathbf{u}_{i}'$ to as the $i$th row of $\mathbf{U}$. Furthermore, we also let $\mathbf{z}_t = [\sqrt{z_{1,t}}, \ldots, \sqrt{z_{L,t}}]^T $ and $\mathbf{g}_t = [\sqrt{g_{1,t}}, \ldots, \sqrt{g_{L,t}}]^T $. Finally, \eqref{eq: Optimization_Coherent} is reformulated as 
\begin{equation} \label{eq: SOCP} 
 \begin{aligned}
 & \underset{ \{ \mathbf{u}_{t} \succeq 0 \}}{\mathrm{minimize}}
 & & \sum_{t=1}^{K} \pmb{\Delta}^T (\mathbf{u}_{t} \circ \mathbf{u}_{t}) \\
 & \mbox{subject to}
 & & || \mathbf{s}_k || \leq \mathbf{g}_k^T \mathbf{u}_k, \; \forall k\\
 & && || \mathbf{u}_{i}'|| \leq \sqrt{P_{\mathrm{max},i}}, \; \forall i. \\
 \end{aligned}
 \end{equation}
Here $\circ$ denotes the element-wise product of two vectors and the vector $\mathbf{s}_k \in \mathbb{C}^{K+ |\mathcal{P}_k|}$ is 
\begin{equation*}
 \left[\sqrt{\hat{\xi}_k} \left( \mathbf{g}_k^T \mathbf{u}_{t_{1}^{'}}, \ldots, \mathbf{u}_{t_{|\mathcal{P}_{k} \setminus \{k\} |}^{'}},  \mathbf{z}_k^T\mathbf{u}_1, \ldots, \mathbf{z}_k^T\mathbf{u}_K \right),\sigma_{\textrm{DL}} \right]^T ,
 \end{equation*}
where $|\cdot|$ denotes the cardinality of a set and $t_{1}^{'}, \ldots, t_{|\mathcal{P}_{k} \setminus \{k\} |}^{'}$ are all the user indices that belong to $\mathcal{P}_k \setminus \{k\}$. We stress that, in \eqref{eq: SOCP}, the object function is convex since it is a quadratic function of variables $\mathbf{u}_t, \forall t$. Additionally, the constraint functions are second-order cones. Consequently, \eqref{eq: SOCP} is a convex program, and therefore the optimal solutions to the power allocation and user association problems can be obtained by using interior-point toolbox CVX \cite{cvx2015}.  Besides, the max-min QoS levels are obtained by solving \eqref{Problem: Max-Min-QoS} in an iterative manner that considers \eqref{eq: SOCP} as the cost function in each iteration.

\bibliographystyle{IEEEtran}
\bibliography{IEEEabrv,refs}

\end{document}